\documentclass[11pt, reqno]{amsart}
\usepackage{amsfonts,latexsym,enumerate}
\usepackage{amsmath}
\usepackage{amscd}
\usepackage{float,amsmath,amssymb,mathrsfs,bm,multirow,graphics}
\usepackage[dvips]{graphicx}
\usepackage[percent]{overpic}
\usepackage[pdftex]{color}
\usepackage{amsaddr}
\usepackage[numbers,sort&compress]{natbib}
\usepackage{enumitem}
\usepackage{todonotes}
\usepackage{mathtools}
\usepackage{comment}

\allowdisplaybreaks

\def\Xint#1{\mathchoice
{\XXint\displaystyle\textstyle{#1}}%
{\XXint\textstyle\scriptstyle{#1}}%
{\XXint\scriptstyle\scriptscriptstyle{#1}}%
{\XXint\scriptscriptstyle\scriptscriptstyle{#1}}%
\!\int}
\def\XXint#1#2#3{{\setbox0=\hbox{$#1{#2#3}{\int}$}
\vcenter{\hbox{$#2#3$}}\kern-.5\wd0}}

\def\pvint{\,\,\Xint-}

\newcommand{\ii}{{\rm i}}
\newcommand{\dd}{{\rm d}}

\newcommand{\im}{\mathrm{Im\,}}

\newcommand{\Res}{\operatorname{Res}}
\newcommand{\res}{\operatorname{Res}}

\newcommand{\R}{{\mathbb R}}
\newcommand{\C}{{\mathbb C}}
\newcommand{\Z}{{\mathbb Z}}

\renewcommand{\Im}{\mathrm{Im}\hspace{0.09em}}

\newcommand{\csch}{\mathrm{csch}}

\newtheorem{theorem}{Theorem}
\newtheorem{proposition}{Proposition}[section]

\newtheorem{lemma}[proposition]{Lemma}
\newtheorem{definition}[proposition]{Definition}

\newtheorem{remark}[proposition]{Remark}

\numberwithin{equation}{section}

\newenvironment{roster}
 {\begin{enumerate}[font=\upshape,label=\Alph*.]}
 {\end{enumerate}}

\begin{document}

\title{On the non-chiral intermediate long wave equation II: periodic case}

\date{\today}

\author{Bjorn K. Berntson$^1$, Edwin Langmann$^2$, and Jonatan Lenells$^1$}
\address{$^1$Department of Mathematics, KTH Royal Institute of Technology, SE-100 44 Stockholm, Sweden \\
$^2$Department of Physics, KTH Royal Institute of Technology, SE-106 91 Stockholm, Sweden}

\begin{abstract}
We study integrability properties of the non-chiral intermediate long wave (ncILW) equation with periodic boundary conditions. The ncILW equation was recently introduced by the authors as a parity-invariant relative of the intermediate long wave equation. For this new equation we: (a) derive a Lax pair, (b) derive a Hirota bilinear form, (c) use the Hirota method to construct the periodic multi-soliton solutions, (d) derive a B\"{a}cklund transformation, (e) use the B\"{a}cklund transformation to obtain an infinite number of conservation laws. 
\end{abstract}

\maketitle



\tableofcontents

%
%
%
%
%

\section{Introduction}
The introduction of the inverse scattering transform for the solution of equations such as the Korteweg-de Vries (KdV), nonlinear Schr\"odinger (NLS), and sine-Gordon equations was a major development in the field of nonlinear PDEs in the 20th century. This development, which began in the late 1960s \cite{gardner1967}, made it clear that certain nonlinear equations, called {\it integrable}, possess unique properties which allow them to be solved exactly, at least in appropriate circumstances. Two classes of solutions of integrable equations are particularly well-studied: (a) the class of solutions on the real line with decay at spatial infinity and (b) the class of (spatially) periodic solutions. These two classes are superficially similar, but they are very different when it comes to details. In fact, throughout the history of integrable PDEs, there has been a fruitful interplay between the theories for these two classes. For example, one of the main tools in the study of solutions on the line is the inverse scattering transform, which provides a way to solve the initial value problem via a sequence of linear operations \cite{ablowitz1991}. The search for a generalization of this approach to the periodic setting led to the introduction of the so-called finite-gap integration method, a development which has in turn influenced the evolution of diverse branches of mathematics as well as theoretical physics; see \cite{matveev2008} for a review. 

In this paper, we consider the periodic version of an integrable equation introduced in \cite{berntson2020a}. This equation is referred to as the non-chiral intermediate long-wave (ncILW) equation, because it involves the same integral operator that appears in the standard intermediate long-wave equation \cite{joseph1977,kodama1981}. However, whereas the latter is chiral in the sense that it only allows for solitons moving in one direction, left or right, the ncILW equation supports solitons moving in both directions.  
While the ncILW equation was discovered in the context of a quantum field theory describing fractional quantum Hall effect systems \cite{berntson2020a}, we expect that it will find applications also in the theory of nonlinear waves and other areas of theoretical physics; see \cite[Sections 1.1--1.3]{berntson2021} for a more detailed discussion of the physics behind the ncILW equation. 

The periodic ncILW equation is given by 
\begin{equation} 
\label{2ilw} 
\begin{split} 
&u_t + 2 u u_x + Tu_{xx}+\tilde{T}v_{xx}=0,\\
&v_t - 2 v v_x - Tv_{xx}-\tilde{T}u_{xx}=0, 
\end{split} 
\end{equation} 
where $u=u(x,t)$ and $v=v(x,t)$ are real- or complex-valued\footnote{For generality, we allow the functions $u$ and $v$ to be complex-valued, but all results can be restricted to the real case without issue.} functions of a space variable $x\in\R$ and a time variable $t\in\R$. The integral operators $T$ and $\tilde{T}$ in \eqref{2ilw} are defined by
\begin{equation}
\label{TT_elliptic}
\begin{split} 
&(Tf)(x) = \frac1{\pi}\pvint_{-L/2}^{L/2} \zeta_1(x'-x|L/2,\ii\delta)f(x')\,\dd{x}',\\
&(\tilde{T}f)(x)=\frac1{\pi}\int_{-L/2}^{L/2} \zeta_1(x'-x+\ii\delta|L/2,\ii\delta)f(x')\,\dd{x}',
\end{split} 
\end{equation}
where 
\begin{equation} 
\label{MLe} 
\zeta_1(z|L/2,\ii\delta) = \frac{\pi}{L}\lim_{M\to\infty}\sum_{n=-M}^M \cot\left(\frac{\pi}{L}(z-2\ii n\delta) \right)
\end{equation} 
is equal, up to a term linear in $z$, to the Weierstrass $\zeta$-function with periods $L>0$ and $2\ii \delta$, $\delta>0$; see Appendix~\ref{app:elliptic} for the precise relation. We are interested in $L$-periodic solutions of this equation, i.e., solutions such that $u(x+L,t)=u(x,t)$ and $v(x+L,t)=v(x,t)$.  The non-chirality of the ncILW equation corresponds to the invariance of \eqref{2ilw} under the parity transformation which maps $x$ to $-x$ and interchanges $u$ and $v$ \cite{berntson2020a}. 

The limiting case $L\to\infty$ where 
\begin{equation*}
\zeta_1(z|L/2,\ii\delta)\to \frac{\pi}{2\delta}\coth\left(\frac{\pi}{2\delta}z\right), \qquad \zeta_1(z+\ii\delta|L/2,\ii\delta)\to \frac{\pi}{2\delta}\tanh\left(\frac{\pi}{2\delta}z\right)
\end{equation*} 
corresponds to the ncILW equation on the real line, i.e., \eqref{2ilw} with the integral operators
\begin{equation}\label{TT}
\begin{split}
&(T_{\R}f)(x) = \frac1{2\delta}\pvint_\R \coth\bigg(\frac{\pi}{2\delta}(x'-x) \bigg)f(x')\,\dd{x}',\\
&(\tilde{T}_{\R}f)(x)=\frac1{2\delta}\int_\R \tanh\bigg(\frac{\pi}{2\delta}(x'-x) \bigg)f(x')\,\dd{x}'. 
\end{split}
\end{equation}
Conversely, to recover \eqref{TT_elliptic} from \eqref{TT}, it is convenient to work in Fourier space; here, the operators in \eqref{TT} have the representation \cite[Eq. A3]{berntson2020a}
\begin{equation}
(\widehat{T_\R f})(k)=\ii \coth(k\delta)\hat{f}(k),\qquad (\widehat{\tilde{T}_\R f})(k)=\ii\, \csch(k\delta)\hat{f}(k),
\end{equation}
where $\csch(z)\coloneqq 1/\sinh(z)$. We assume $f(x)$ is a zero-mean $L$-periodic function with the Fourier transform pair:
\begin{equation}\label{Fourierpair}
f(x)=\sum_{n\in \Z\setminus\{0\}} \hat{f}_n e^{2\ii\pi n x/L},\qquad  \hat{f}(k)=2\pi\sum_{n\in \Z\setminus\{0\}} \hat{f}_n \delta(k-2\pi n/L).
\end{equation}
It follows that
\begin{align}\label{TTFourier}
\begin{split}
({T_\R f})(x)=&\; \ii \sum_{n\in \Z\setminus\{0\}} \coth\bigg(\frac{2n\pi \delta}{L}\bigg)\hat{f}_n e^{2\ii\pi n x/L},\\
 ({\tilde{T}_\R f})(x)=&\; \ii \sum_{n\in \Z\setminus\{0\}} \csch \bigg(\frac{2n\pi \delta}{L}\bigg)\hat{f}_n e^{2\ii\pi n x/L}.
 \end{split}
\end{align}
Comparing \eqref{TTFourier} with the Fourier series for the functions $\zeta_1(z)$ and $\zeta_1(z+\ii\delta)$ \cite[Eq. 23.8.2]{DLMF} appearing in \eqref{TT}, it is straightforward to show (at least formally) \cite{ablowitz1982} that $T_{\R}=T$ and $\tilde{T}_\R=\tilde{T}$ on such functions \eqref{Fourierpair} (the functions $u_{xx},v_{xx}$ appearing as arguments of $T,\tilde{T}$ in \eqref{2ilw} are in this class).

In a recent paper \cite{berntson2021}, we obtained a Lax pair, a Hirota form, a B\"acklund transformations and an infinite number of conservation laws for the ncILW equation on the real line. 
In this paper, we present corresponding results in the periodic case. 
We show that, even though several steps are significantly more complicated in the periodic setting, it is nevertheless possible to prove results which parallel those obtained in \cite{berntson2021}. 
We also provide an alternative derivation based on the Hirota method of the multi-soliton solutions of the periodic ncILW equation; the multi-solitons were previously obtained by a different method in \cite{berntson2020a}. 
More precisely, we show that the multi-solitons can be obtained via a pole ansatz in terms of a Weierstrass $\zeta$-function, where the poles evolve according to the elliptic Calogero-Moser (CM) system; see Proposition~\ref{solitoncorollary_elliptic} for the precise formulation. 

In the limit $L\to \infty$, the results we obtain here for the periodic problem reduce (at least formally) to analogous results for the problem on the line. However, we emphasize that only a subset of the results of \cite{berntson2021} can be obtained in this way: proving the results directly on the real line is not only technically simpler but also leads to more general results. 

It is important to note that the non-chiral ILW equation \eqref{2ilw} is an {\it elliptic} integrable systems. This is clear already from the definition of the operators $T$ and $\tilde{T}$ in \eqref{TT_elliptic}: while the ncILW equation on the line involves integral operators whose kernels are given in terms of the hyperbolic tangent, the kernels in \eqref{TT_elliptic} involve a Weierstrass $\zeta$-function. 
The fact that \eqref{2ilw} is an elliptic system is also evident from the relation to the elliptic CM system mentioned above. In fact, the periodic ncILW equation is related to the elliptic CM system in the same way as the ncILW equation on the line is related to the hyperbolic CM system \cite{berntson2020a}. 

The plan of this paper is as follows. In Section~\ref{laxpairsec}, we derive a Lax pair for \eqref{2ilw}. A Hirota bilinear form is presented in Section \ref{hirotasec}, where we additionally prove that the Hirota bilinear form is equivalent to \eqref{2ilw} by constructing explicit transformations from $(u,v)$ to the Hirota variables $(F,G)$ and vice-versa.  We use the Hirota bilinear form to construct $N$-periodic soliton solutions via a pole ansatz in Section~\ref{solitonsec}. A B\"acklund transformation is constructed from the Hirota bilinear form in Section \ref{backlundsec}.
Definitions and basic properties of certain elliptic functions are collected in Appendix~\ref{app:elliptic}. Some properties of the operators $T$ and $\tilde{T}$ defined in \eqref{TT_elliptic} are established in Appendix~\ref{app:TT}.

In what follows we assume that the arguments of $T$ and $\tilde{T}$ are sufficiently regular to justify our arguments. We occasionally comment on specific necessary or sufficient conditions for clarity.

\section{Lax pair}\label{laxpairsec}

We will obtain a Lax pair for \eqref{2ilw} with \eqref{TT_elliptic} via a Riemann-Hilbert (RH) problem with two jumps on a torus. We construct this torus as $\Pi=\C/\Lambda$, where $\Lambda\coloneqq L \Z+2\ii\delta \Z$. Let $\pi:\C\to\Pi$ be the natural projection; we identify $\Pi$ with the parallelogram
\begin{equation*}
\Pi \cong \{(L/2)r+\ii\delta s:-1<r,s\leq 1\}.
\end{equation*}
A function $f:\Pi\to\C$ can be viewed as a function $f:\C\to\C$ which is doubly periodic with periods $L$ and $2\ii\delta$, i.e.
\begin{equation*}
f(z+mL+2\ii n\delta)=f(z),\qquad m,n \in \Z.
\end{equation*}
Let $\Pi_0$ and $\Pi_\delta$ denote the images of the lines $\Im z=0$ and $\Im z=\delta$, respectively, under $\pi$. We consider an eigenfunction $\psi(z,t;k)$; for each $t\in\R$ and $k\in\C$, $\psi(z)\coloneqq\psi(z,t;k)$ is an analytic function $\Pi\setminus(\Pi_0\cup\Pi_\delta)$ with jumps across $\Pi_0$ and $\Pi_\delta$. The boundary values of the eigenfunction are functions $\Pi_0\cup \Pi_\delta\to\C$ defined by 
\begin{align}\label{psi_bv}
\psi^{\pm}(x,t;k)\coloneqq\lim\limits_{\epsilon \downarrow 0} \psi(x\pm\ii\epsilon,t;k),\qquad \psi^{\pm}(x+\ii\delta,t;k)\coloneqq\lim\limits_{\epsilon \downarrow 0} \psi(x+\ii\delta\pm\ii\epsilon,t;k).
\end{align}
We take the following ansatz for the Lax pair:
\begin{equation}\label{lax_ansatz_elliptic}
\begin{cases}
\ii \psi^-_x+(-u-\mu_1)\psi^-=\nu_1\psi^+ \qquad &\text{for } z\in \Pi_0, \\
\ii \psi^+_x+(v-\mu_2)\psi^+=\nu_2\psi^- \qquad &\text{for } z\in \Pi_\delta, \\
\psi_t+\ii \psi_{xx}-\ii A(z,t;k)\psi-\ii B(z,t;k)\psi_x=0 \qquad & \text{for } z\in \Pi \setminus(\Pi_0\cup \Pi_\delta).
\end{cases}
\end{equation}
Here, $\mu_1$, $\mu_2$, $\nu_1$, and $\nu_2$ are complex-valued functions of the spectral parameter; $A(z,t;k)$ and $B(z,t;k)$ are analytic functions on $\Pi\setminus(\Pi_0\cup \Pi_\delta)$ to be determined. To obtain the compatibility conditions for \eqref{lax_ansatz_elliptic}, we write the boundary values of the $t$-part of the Lax pair:
\begin{equation*}
\psi^{\pm}_t+\ii \psi^{\pm}_{xx}-\ii A^{\pm}(z,t;k)\psi-\ii B^{\pm}(z,t;k)\psi^{\pm}_x=0, \qquad \text{for } z\in \Pi_0\cup \Pi_\delta.
\end{equation*}
This equation and its $x$-derivative can be used to eliminate $\psi^{\pm}_t$ and $\psi^{\pm}_{tx}$ from the $t$-derivative of the $x$-part of \eqref{lax_ansatz_elliptic}, leading to 
\begin{align*}
& \ii \nu_1(B^+ - B^-) \psi_x^+ 
+ \ii\nu_1(A^- - A^+ + B_x^- +  2\ii u_x) \psi^+
	\\
&+ \big[-u_t - A_x^- + \ii (\mu_1 + u)B_x^- - 2\mu_1 u_x + \ii B^- u_x - 2uu_x - \ii u_{xx}\big]\psi^+ =  0, \qquad \text{on $\Pi_0$},
\end{align*}
and
\begin{align*}
& \ii \nu_2(B^+ - B^-) \psi_x^+ 
+ \ii\nu_2(A^+ - A^- + B_x^+ - 2\ii v_x) \psi^+
	\\
&+ \big[v_t - A_x^+ + \ii (\mu_2 - v)B_x^+ + 2\mu_2 v_x - \ii B^+ v_x - 2vv_x + \ii v_{xx}\big]\psi^+ =  0, \qquad \text{on $\Pi_\delta$}.
\end{align*}
Setting the coefficients of $\psi_x^\pm$ and $\psi^\pm$ to zero, we find the equations
\begin{subequations}
\begin{align}\label{biconda_elliptic}
& B^+ - B^- = 0 \qquad \text{on $\Pi_0 \cup \Pi_\delta$},
	\\\label{bicondb2_elliptic}
& A^+ -A^- - B_x^- -  2\ii u_x = 0 \qquad \text{on $\Pi_0$},
	\\\label{bicondc2_elliptic}
& A^+ - A^- + B_x^+ - 2\ii v_x = 0 \qquad \text{on $\Pi_\delta$},
 	\\\label{biconde2_elliptic}
& u_t + A_x^- - \ii(\lambda_1 + u)B_x^- + 2\mu_1 u_x - \ii B^- u_x + 2uu_x + \ii u_{xx} = 0 \qquad \text{on $\Pi_0$},
	\\\label{bicondd2_elliptic}
& v_t - A_x^+ + \ii(\lambda_2 -v)B_x^+ + 2\mu_2 v_x - \ii B^+ v_x - 2vv_x + \ii v_{xx} =  0 \qquad \text{on $\Pi_\delta$}.
\end{align}
\end{subequations}
From \eqref{bicondb2_elliptic}, we see that $B(z)$ is an analytic function on $\Pi$ and so must be constant: $B(z)=B_0$. Then (\ref{bicondc2_elliptic}-\ref{biconde2_elliptic}) shows that $A$ is a solution of the following RH problem on $\Pi$:
\begin{itemize}
\item $A:\Pi \setminus (\Pi_0 \cup \Pi_\delta) \to \C$ is an analytic function,
\item across $\Pi_0 \cup \Pi_\delta$, $A$ satisfies the jump condition
$$A^+(z) - A^-(z) = \begin{cases} 2\ii u_x(x), \qquad z = x \in \Pi_0, \\
2\ii v_x(x), \qquad z = x+ \ii\delta \in \Pi_\delta.
\end{cases}$$
\end{itemize}

\begin{lemma}[RH problem on $\Pi$ with a jump across $\Pi_0 \cup \Pi_\delta$]\label{RHlemma_elliptic}
Let $J_0: \Pi_0 \to \C$ and $J_1: \Pi_\delta \to \C$   be continuous functions satisfying
$$\int_{-L/2}^{L/2} J_0(x)\,\mathrm{d}x = \int_{-L/2}^{L/2} J_1(x)\,\mathrm{d}x = 0.$$ 
Define $J: \Pi_0 \cup \Pi_\delta \to \C$ by
$$J(z) = \begin{cases} J_0(x), & z \in \Pi_0, \\
J_1(x), & z \in \Pi_\delta.
\end{cases}$$
Then the scalar RH problem: 
\begin{itemize}
\item $A:\Pi \setminus (\Pi_0 \cup \Pi_\delta)$ is analytic, 
\item across $\Pi_0 \cup \Pi_\delta$, $A$ satisfies the jump condition, 
$$A^+(z) - A^-(z) = J(z), \qquad z \in \Pi_0 \cup \Pi_\delta$$
\end{itemize}
has the general solution
\begin{align}\label{RHsolution2_elliptic}
A(z) = \frac{1}{2\pi\ii} \int_{\Pi_0 \cup \Pi_\delta} \zeta_1(z'-z|L/2,\ii\delta) J(z')\,\mathrm{d}z'+A_0, \qquad z \in \Pi \setminus (\Pi_0 \cup \Pi_\delta),
\end{align}
where $A_0$ is an arbitrary complex constant and both $\Pi_0$ and $\Pi_\delta$ are oriented from $\im z=-L/2$ to $\im z=L/2$. Moreover, this solution satisfies
\begin{equation}\label{RHsolution2_bv_elliptic}
A^\pm(z) = \begin{cases} 
\frac{(TJ_0)(x) + (\tilde{T} J_1)(x)}{2\mathrm{i}} \pm \frac{1}{2} J_0(x), & z = x \in \Pi_0,
	\\
\frac{(\tilde{T} J_0)(x) + (TJ_1)(x)}{2\mathrm{i}} \pm \frac{1}{2} J_1(x), & z = x + \ii \delta \in \Pi_\delta.
\end{cases}
\end{equation}
\end{lemma}
\begin{proof}
If $A_1$ and $A_2$ are two different solutions, then $A_1 - A_2$ is analytic on $\Pi$ and hence constant. Let $A$ be given by (\ref{RHsolution2_elliptic}).  Using periodicity properties of $\zeta_1$, we observe that $A(z+L)=A(z)$ and 
\begin{equation}
A(z+2\mathrm{i}\delta)=A(z)-\frac1{2L}\int_{\Pi_0\cup\Pi_{\delta}} J(z)\,\mathrm{d}z,
\end{equation}
where the integral vanishes by assumption. Hence $A$ descends to a well-defined function $A:\Pi\setminus(\Pi_0\cup\Pi_{\delta})\rightarrow\C$.

For $x \in \Pi_0$, the Plemelj formula gives (we suppress the second and third arguments of $\zeta_1$)
\begin{align*}
A^\pm(x)-A_0 
= &\;\frac{1}{2\pi\mathrm{i}} \bigg\{\pvint_{\Pi_0}\zeta_1(z'-x) J(z')\,\mathrm{d}z'
+ \int_{\Pi_\delta} \zeta_1(z'-x) J(z')\,\mathrm{d}z'
	\\
& \pm \pi \mathrm{i}\, \underset{z' = x}\Res\, \zeta_1(z'-x) J_0(x)\bigg\}
	\\
= &\; \frac{1}{2\mathrm{i}} (TJ_0)(x) 
+ \frac{1}{2\pi\mathrm{i}} \int_{\Pi_0} \zeta_1(x'-x+\mathrm{i}\delta) J_1(x') \,\mathrm{d}x'
\pm \frac{1}{2} J_0(x)
	\\
= &\; \frac{1}{2\mathrm{i}} (TJ_0)(x) + \frac{1}{2\mathrm{i}} (\tilde{T}J_1)(x) \pm \frac{1}{2} J_0(x).
\end{align*}

Similarly, for $x + \mathrm{i}\delta \in \Pi_\delta$,
\begin{align*}
A^\pm(x + \mathrm{i}\delta)-A_0
= &\; \frac{1}{2\pi\mathrm{i}} \bigg\{\int_{\Pi_0} \zeta_1(z'-x-\mathrm{i}\delta) J(z')\, \mathrm{d}z'+ \pvint_{\Pi_{\delta}} \zeta_1(z'-x-\mathrm{i}\delta) J(z')\, \mathrm{d}z'
	\\
& \pm \pi \mathrm{i} \underset{z' = x+\ii \delta}\Res \zeta_1(z'-x-\mathrm{i}\delta) J(z')\bigg\}
	\\
= &\; \frac{1}{2\mathrm{i}} (\tilde{T}J_0)(x) 
+ \frac{1}{2\pi\mathrm{i}} \pvint_{\Pi_0} \zeta_1(x'-x) J_1(x') \,\mathrm{d}x'
\pm \frac{1}{2} J_1(x) \\
= &\; \frac{1}{2\mathrm{i}} (\tilde{T}J_0)(x) + \frac{1}{2\mathrm{i}} (TJ_1)(x) \pm \frac{1}{2} J_1(x).
\end{align*}
This proves the expressions for the boundary values and shows that $A$ satisfies the correct jump condition. 
\end{proof}

Using \eqref{RHsolution2_elliptic}, we see that
\begin{align}\label{Asolution_elliptic}
A(z,t;k)=&\; \frac{1}{\pi}\int_{\Pi_0} \zeta_1(z'-z|L/2,\ii\delta)J_0(z')\mathrm{d}z' \\
&+\frac{1}{\pi}\int_{\Pi_\delta} \zeta_1(z'-z|L/2,\ii\delta)J_1(z')\mathrm{d}z'+A_0(k),\qquad z\in \Pi\setminus(\Pi_0\cup\Pi_\delta), \nonumber
\end{align}
and
\begin{equation}\label{RHsolution2_bv_elliptic_2}
A^\pm(z,t;k) = \begin{cases} 
(TJ_0)(x) + (\tilde{T} J_1)(x) \pm \ii u(x)+A_0(k), & z = x \in \Pi_0,
	\\
(\tilde{T} J_0)(x) + (TJ_1)(x) \pm \ii v(x)+A_0(k), & z = x + \ii \delta \in \Pi_\delta.
\end{cases}
\end{equation}
Substituting these expressions for $A^\pm$ into \eqref{bicondd2_elliptic} and \eqref{biconde2_elliptic} and using that $T$ and $\tilde{T}$ commute with $\partial_x$ from Proposition \ref{TpropertiesS}, we arrive at the two-component equation
\begin{align*}
& u_t + Tu_{xx} + \tilde{T}v_{xx} + 2\mu_1 u_x - \ii B_0 u_x + 2uu_x  = 0,
	\\
& v_t - Tv_{xx} - \tilde{T}u_{xx} + 2\mu_2 v_x - \ii B_0 v_x - 2vv_x =  0.
\end{align*}
Choosing $\mu_1 = \mu_2 = \mu$ and $B_0 = -2\ii\mu$, this becomes the non-chiral ILW equation \eqref{2ilw}. 
We summarize the results above in a theorem.

\begin{theorem}[Lax pair for the periodic ncILW equation]
The periodic ncILW equation is the compatibility condition of the Lax pair
\begin{align}\label{laxpair_elliptic}
\begin{cases}
i\psi_x^- + (-u-\mu)\psi^- = \nu_1 \psi^+ & \text{on $\Pi_0$},
	\\
i\psi_x^+ + (v-\mu)\psi^+ = \nu_2 \psi^- & \text{on $\Pi_\delta$},
	\\
\psi_t^\pm + \ii \psi_{xx}^\pm - 2\mu \psi_x^\pm - \ii (Tu_x + \tilde{T}v_x \pm \ii u_x + A_0) \psi^\pm=0 & \text{on $\Pi_0$},
	\\
\psi_t^\pm + \ii \psi_{xx}^\pm - 2\mu \psi_x^\pm - \ii (Tv_x  + \tilde{T}u_x  \pm \ii v_x + A_0) \psi^\pm=0 & \text{on $\Pi_\delta$},
\end{cases}
\end{align}
where $\mu=\mu(k)$, $\nu_1=\nu_1(k)$, $\nu_2=\nu_2(k)$, and $A_0=A_0(k)$ are complex parameters which may depend on the spectral parameter $k$.
\end{theorem}

\begin{remark}
The $t$-parts of \eqref{laxpair_elliptic} have an analytic continuation to $\Pi\setminus(\Pi_0\cup\Pi_\delta)$ and can be alternatively written as
\begin{equation*}
\psi_t+\ii\psi_{zz}-2\mu\psi_z-\ii A \psi=0,\qquad z\in \Pi\setminus(\Pi_0\cup\Pi_\delta),
\end{equation*}
where $A=A(z,t;k)$ is given by \eqref{Asolution_elliptic}. 
\end{remark}

\section{Hirota bilinear form} \label{hirotasec}

In the periodic setting, the Hirota bilinear form of \eqref{2ilw} is 
\begin{subequations}\label{hirota_form_elliptic}
\begin{align}
&\big(\ii D_t-D_x^2+2\ii \bar{u} D_x-\lambda_1(t)+\bar{u}^2\big)F^-\cdot G^+=0, \label{hirota_form_a_elliptic}\\
&\big(\ii D_t-D_x^2-2\ii \bar{v}D_x-\lambda_2(t)+\bar{v}^2\big)F^+\cdot G^-=0,\label{hirota_form_b_elliptic}
\end{align}
\end{subequations}
where $\bar{u}$ and $\bar{v}$ are the spatial means of $u(x,t)$ and $v(x,t)$, respectively and $\lambda_1(t)$ and $\lambda_2(t)$ are complex functions. By the following lemma, we may take $\bar{u}$ and $\bar{v}$ to be constants. 

\begin{lemma}\label{uvmeans}
The means of $u$ and $v$ in \eqref{2ilw} are independent of time. 
\end{lemma}
\begin{proof}
Starting from the definition of the mean
\begin{equation*}
\bar{u}(t)\coloneqq\frac1L\int_{-L/2}^{L/2} u(x,t)\,\mathrm{d}x,
\end{equation*}
we compute 
\begin{align*}
\bar{u}_t=&\frac1L\int_{-L/2}^{L/2} u_t\,\mathrm{d}x=-\frac{1}{L}\int_{-L/2}^{L/2} (2uu_x+Tu_{xx}-\tilde{T}v_{xx})\,\mathrm{d}x\\
=&-\frac1L\big[ u^2+Tu_x-\tilde{T}v_x   \big]^{L/2}_{-L/2}=0.
\end{align*}
The proof for $\bar{v}$ is similar. 
\end{proof}

A bilinear form similar to \eqref{hirota_form_elliptic} was used in \cite{parker1992} to construct periodic solutions of the standard ILW equation. We show that \eqref{hirota_form_elliptic} is equivalent to \eqref{2ilw} in the sense of the following theorem.

\begin{theorem}[Hirota bilinear form of periodic non-chiral ILW]\label{bilinearth_elliptic}
\hfill \break
\begin{roster}\item
Let $F(z,t)$ and $G(z,t)$ be $L$-periodic functions of $z\in \C$ and $t\in \R$ such that $\log F(z,t)$ and $\log G(z,t)$ are analytic for $-\delta/2<\im z<\delta/2$ and continuous for $-\delta/2\leq\im z\leq \delta/2$. Then $F,G$ satisfy the bilinear system \eqref{hirota_form_elliptic} for some $\bar{u},\bar{v}\in\C$ and complex-valued functions $\lambda_1(t)$ and $\lambda_2(t)$ if and only if
\begin{equation}\label{FG_to_uv_elliptic}
u=\bar{u}+\ii \partial_x\log \frac{F^-}{G^+},\qquad v=\bar{v}+\ii\partial_x\log \frac{G^-}{F^+}.
\end{equation}
satisfy \eqref{2ilw}.
\item 
Suppose $u(x,t)$, $v(x,t)$ are $L$-periodic solutions of \eqref{2ilw} with means $\bar{u}$ and $\bar{v},$ respectively. Then $F(x,t)$, $G(x,t)$, defined up to multiplication by an arbitrary function of $t$ by 
\begin{equation}
\begin{split}\label{FG_from_uv_elliptic}
\begin{dcases}
\ii \partial_z \log F(z,t)= \frac1{\ii \pi}\int_{-L/2}^{L/2}    \zeta_1(x'-z|L/2,\ii\delta/2)   (u_{+}(x')+v_{+}(x')\big)\,\mathrm{d}x', \\
\ii\partial_z \log G(z,t)=\frac{1}{\ii\pi}  \int_{-L/2}^{L/2} \zeta_1(x'-z|L/2,\ii\delta/2)\big(u_{-}(x')+v_{-}(x')\big)\,\mathrm{d}x',
\end{dcases}
\end{split}
\end{equation}
where
\begin{equation}
\begin{split}\label{uvprojections_elliptic}
\begin{dcases}
u_{\pm}\coloneqq\frac12 (u-\bar{u})\mp\frac{\ii}{2}\big(T(u-\bar{u})+\tilde{T}(v-\bar{v})\big) \\
v_{\pm}\coloneqq\frac12 (v-\bar{v})\pm\frac{\ii}{2}\big(T(v-\bar{v})+\tilde{T}(u-\bar{u})\big)
\end{dcases},\qquad x,t\in\R,
\end{split}
\end{equation}
are analytic for $-\delta/2<\im z<\delta/2$, continuous for $-\delta/2\leq \im z\leq \delta/2$, and
satisfy the Hirota equations \eqref{hirota_form_elliptic} with
\begin{equation}\label{lambdavalues}
\begin{cases}
\lambda_1=\ii (\log F^-/G^+)_t+u^2+\ii(u_+-u_-)_x ,\\
\lambda_2=\ii(\log F^+/G^-)_t+v^2-\ii(v_+-v_-)_x .
\end{cases}
\end{equation}
\end{roster}
\end{theorem}

\subsection{Proof of Theorem \ref{bilinearth_elliptic}A}
Suppose $(F,G)$ and $(u,v)$ are related as in \eqref{FG_to_uv_elliptic}. We write
\begin{equation*}
u=\bar{u}+u_++u_-,\qquad v=\bar{v}+v_++v_-,
\end{equation*}
where $u_{\pm}$ and $v_{\pm}$ are defined by
\begin{equation}\label{uvplusminusdef_elliptic}
\begin{split}
u_+(z,t)\coloneqq&\; \ii\partial_z\log F(z-\ii\delta/2,t),\qquad  u_-(z,t)\coloneqq-\ii\partial_z\log G(z+\ii\delta/2,t), \\
v_+(z,t)\coloneqq&-\ii\partial_z\log F(z+\ii\delta/2,t),\qquad v_-(z,t)\coloneqq\; \ii\partial_z\log G(z-\ii\delta/2,t).
\end{split}
\end{equation}
Each of these functions has zero mean by $L$-periodicity of $F$ and $G$. By our assumptions on the analyticity of $\log F$ and $\log G$, we see that $u_+$ and $v_-$ are analytic in the strip $0< \im z<\delta$, $u_-$ and $v_+$ are analytic in the strip $-\delta<\im z< 0$. Additionally, we observe that
\begin{equation}\label{uvplusminusrelations_elliptic}
v_+(z,t)=-u_+(z+\ii\delta,t),\qquad v_-(z,t)=-u_-(z-\ii\delta,t).
\end{equation}

\begin{lemma}\label{lemma1_elliptic}
If $g^+(z)$ is $L$-periodic, analytic in the strip $0 < \im z < \delta$, and continuous in the strip $0 \leq \im z \leq \delta$,
then,
\begin{align}\label{TTtildegplus_elliptic}
(Tg^+)(x) - (\tilde{T}[g^+(\cdot + \ii\delta)])(x) = \ii g^+(x), \qquad x \in \R.
\end{align}
Similarly, if $g^-(z)$ is $L$-periodic, analytic in the strip $-\delta < \im z < 0$, and continuous in the strip $-\delta \leq \im z \leq 0$, 
then
\begin{align}\label{TTtildegminus_elliptic}
(Tg^-)(x) - (\tilde{T}[g^-(\cdot - \ii\delta)])(x) = -\ii g^-(x)+\frac{2\ii}{L}\int_{-L/2}^{L/2} g^-(x'-\ii\delta) \,\mathrm{d}x', \qquad x \in \R.
\end{align}
\end{lemma}
\begin{proof}
Suppose $g^+(z)$ is an $L$-periodic function which is analytic in $0 < \im z < \delta$, and continuous in $0 \leq \im z \leq \delta$. Using the definition of $\tilde{T}$ \eqref{TT_elliptic} and then changing variables to $z' = x' +\ii\delta$, we find
\begin{align*}
(\tilde{T}[g^+(\cdot + \ii\delta)])(x) 
& = \frac{1}{\pi} \int_{-L/2}^{L/2} \zeta_1(x'-x+\ii\delta) g^+(x' + \ii\delta)\,\mathrm{d}x'
  	\\
& = \frac{1}{\pi} \int_{-L/2+\ii\delta}^{L/2+\ii\delta} \zeta_1(z'-x) g^+(z')\,\mathrm{d}z'.
\end{align*}
We deform the contour down towards the real axis. Utilizing the Plemelj formula to evaluate the contribution from the simple pole at $z' = x$, we obtain
\begin{align}\nonumber
(\tilde{T}[g^+(\cdot + \ii\delta)])(x) 
= &\; \frac{1}{\pi} \pvint_{-L/2}^{L/2} \zeta_1(z'-x) g^+(z')\,\mathrm{d}z' \\
& - \ii\,\underset{z'=x}\res\, \zeta_1(z'-x) g^+(z') \label{tildeTgpluscdot_elliptic} 
+ E(x),
\end{align}
where
$$E(x) \coloneqq \frac{1}{\pi}
\bigg(\int_{-L/2+ \ii\delta}^{-L/2} + \int_{L/2}^{L/2+ \ii\delta} \bigg)\zeta_1(z'-x) g^+(z')\, \mathrm{d}z'.$$
Because the integrand of $E$ is $L$-periodic, we have $E=0$. Thus, using that $\underset{z' = x}{\res} \,\zeta_1(z'-x) = 1, $
equation (\ref{tildeTgpluscdot_elliptic}) reduces to 
\begin{align}\nonumber
(\tilde{T}[g^+(\cdot + \ii\delta)])(x) = (T g^+)(x) - \ii g^+(x),
\end{align}
which is (\ref{TTtildegplus_elliptic}). 

The proof of (\ref{TTtildegminus_elliptic}) is similar, but there is a correction term due to the non-$2\ii\delta$-periodicity of $\zeta_1$. Suppose $g^-(z)$ is an $L$-periodic function which is analytic in $-\delta<\im z<0$ and continuous in $-\delta\leq\im z\leq 0$. Using the definition of $\tilde{T}$, changing variables to $z' = x' +\ii\delta$, and using the identity $\zeta_1(z-2\omega_2)=\zeta_1(z)+\ii\pi/\omega_1$ from Proposition \ref{TpropertiesS}, we find
\begin{align*}
(\tilde{T}[g^-(\cdot - \ii\delta)])(x) 
& = \frac{1}{\pi} \int_{-L/2}^{L/2} \zeta_1(x'-x+\ii\delta) g^-(x' - \ii\delta)\,\mathrm{d}x'
  	\\
& = \frac{1}{\pi} \int_{-L/2-\ii\delta}^{L/2-\ii\delta} \zeta_1(z'-x-2\ii\delta) g^-(z')\,\mathrm{d}z' \\
&= \frac{1}{\pi} \int_{-L/2-\ii\delta}^{L/2-\ii\delta} \zeta_1(z'-x) g^-(z')\,\mathrm{d}z'-\frac{2\ii}{L}\int_{-L/2-\ii\delta}^{L/2-\ii\delta} g^-(z') \,\mathrm{d}z' \\
&= \frac{1}{\pi} \int_{-L/2-\ii\delta}^{L/2-\ii\delta} \zeta_1(z'-x) g^-(z')\,\mathrm{d}z'-\frac{2\ii}{L}\int_{-L/2}^{L/2} g^-(x'-\ii\delta) \,\mathrm{d}x'. 
\end{align*}
We deform the contour up towards the real axis. Utilizing the Plemelj formula to evaluate the contribution from the simple pole at $z' = x$, we obtain
\begin{align}\nonumber
(\tilde{T}[g^-(\cdot - \ii\delta)])(x) 
= &\; \frac{1}{\pi} \pvint_{-L/2}^{L/2} \zeta_1(z'-x) g^-(z')\,\mathrm{d}z' \\
& + \ii\,\underset{z'=x}\res\, \zeta_1(z'-x) g^-(z')  \nonumber \\
&-\frac{2\ii}{L}\int_{-L/2}^{L/2} g^-(x'-\ii\delta) \,\mathrm{d}x'  + E(x) \label{tildeTgminuscdot_elliptic}
\end{align}
where
$$E(x) \coloneqq \frac{1}{\pi}
\bigg(\int_{-L/2- \ii\delta}^{-L/2} + \int_{L/2}^{L/2- \ii\delta} \bigg)\zeta_1(z'-x) g^-(z') \,\mathrm{d}z'.$$
Because the integrand of $E$ is $L$-periodic, we have $E=0$. Thus, using that $\underset{z' = x}{\res} \,\zeta_1(z'-x) = 1, $
equation (\ref{tildeTgminuscdot_elliptic}) reduces to 
\begin{align}\nonumber
(\tilde{T}[g^-(\cdot + \ii\delta)])(x) = (T g^-)(x) + \ii g^-(x)-\frac{2\ii}{L}\int_{-L/2}^{L/2} g^-(x'-\ii\delta) \,\mathrm{d}x',
\end{align}
which is (\ref{TTtildegminus_elliptic}). 
\end{proof}

\begin{lemma}\label{lemma2_elliptic}
The functions $u$ and $v$ obey the identities
$$  \begin{cases}
  Tu + \tilde{T}v = \ii(\bar{v}+u_+ - u_-), \\
  Tv + \tilde{T}u = \ii(\bar{u}-v_+ + v_-).
\end{cases} \qquad x,t \in \R.$$
\end{lemma}
\begin{proof}

By (\ref{uvplusminusrelations_elliptic}) and the identities $T[1]=0$ and $\tilde{T}[1]=-\ii$ from Proposition \ref{TpropertiesS}, we have
\begin{align}\label{TutildeTv_elliptic}
Tu + \tilde{T}v = &\; T(\bar{u}+u_+ + u_-) + \tilde{T}(\bar{v}+v_+ + v_-) \\
= &\; Tu_+ - \tilde{T}[u_+(\cdot + \ii\delta)] + Tu_- - \tilde{T}[u_-(\cdot - \ii\delta)]+T\bar{u}-\tilde{T}\bar{v} \nonumber  \\
= &\; Tu_+ - \tilde{T}[u_+(\cdot + \ii\delta)] + Tu_- - \tilde{T}[u_-(\cdot - \ii\delta)]+\ii \bar{v}. \nonumber
\end{align}
We see from (\ref{uvplusminusdef_elliptic}) that $u_+$ is $L$-periodic, analytic for $0 < \im z < \delta$, and continuous for $0 \leq \im z \leq \delta$, so that $Tu_+-\tilde{T}[u_+(\cdot+\ii\delta)]=\ii u_+$ by Lemma \ref{lemma1_elliptic}. Similarly, $u_-$ is $L$-periodic, analytic for $-\delta < \im z < 0$, continuous for $-\delta \leq \im z \leq 0$, and has zero mean, so that $Tu_--\tilde{T}[u_-(\cdot-\ii\delta)]=-\ii u_-$. Hence, the identity $Tu + \tilde{T}v = \ii(\bar{v}+u_+ - u_-)$ follows from (\ref{TutildeTv_elliptic}) and Lemma \ref{lemma1_elliptic}. 

The proof of the identity $Tv + \tilde{T}u = \ii(\bar{u}-v_+ + v_-)$ is similar. Indeed, by (\ref{uvplusminusrelations_elliptic}) and (\ref{Tconstantint1}-\ref{Ttconstantint1}), we have
\begin{align*}
Tv + \tilde{T}u =& \; T(v_+ + v_-) + \tilde{T}(u_+ + u_-) \\
= &\; Tv_+ - \tilde{T}[v_+(\cdot - \ii\delta)] + Tv_- - \tilde{T}[v_-(\cdot + \ii\delta)]+T\bar{v}-\tilde{T}\bar{u} \nonumber \\
= &\; Tv_+ - \tilde{T}[v_+(\cdot - \ii\delta)] + Tv_- - \tilde{T}[v_-(\cdot + \ii\delta)]+\ii \bar{u}. \nonumber 
\end{align*}
By our analyticity assumptions on $F,G$, we see from (\ref{uvplusminusdef_elliptic}) that $v_-$ is analytic for $0 < \im z < \delta$, and continuous for $0 \leq \im z \leq \delta$, so that $Tv_--\tilde{T}[v_-(\cdot+\ii\delta)]=\ii v_-$ by Lemma \ref{lemma1_elliptic}. Similarly, $v_+$ is analytic for $-\delta<\im z<0$, continuous for $-\delta\leq \im z \leq 0$, and has zero mean, so that $Tv_+-\tilde{T}[v_+(\cdot-\ii\delta)]=-\ii v_+$. Hence, the identity $Tv + \tilde{T}u = \ii(\bar{u}-v_+ + v_-)$ follows from (\ref{TutildeTv_elliptic}) and Lemma \ref{lemma1_elliptic}. 
\end{proof}

\begin{proof}[Proof of Theorem \ref{bilinearth_elliptic}A]
According to Lemma \ref{lemma2_elliptic}, the non-chiral ILW equation (\ref{2ilw}) can be written as
\begin{align}\label{2ILW2_elliptic}
\begin{cases}
u_t  + 2uu_x + \ii(u_+ - u_-)_{xx} =  0,
	\\
v_t  - 2vv_x + \ii(v_+ - v_-)_{xx} =  0.
\end{cases}
\end{align}
Since $u_t =\ii(\log {{F}^-}/{{G}^+})_{xt}$ and $v_t = \ii(\log {{G}^-}/{{F}^+})_{xt}$, integration of (\ref{2ILW2_elliptic}) with respect to $x$ gives
\begin{align}\label{hirota_integrated}
\begin{cases}
\ii\big(\log {F}^-/{G}^+\big)_{t}  + u^2 + \ii(u_+ - u_-)_{x} = \lambda_1(t),
	\\
\ii\big(\log {{G}^-}/{{F}^+} \big)_{t}  - v^2 + \ii(v_+ - v_-)_{x} = -\lambda_2(t),
\end{cases}
\end{align}
where $\lambda_1(t)$ and  $\lambda_2(t)$ are arbitrary complex functions. 

Rewriting the system in terms of $F$ and $G$, we obtain
\begin{align*}
\begin{cases}
\ii\big(\frac{{F}^-_t}{{F}^-} - \frac{{G}^+_t}{{G}^+}\big) + \big(\bar{u}(t)+\ii\frac{{F}^-_x}{{F}^-} - \ii\frac{{G}^+_x}{{G}^+}\big) ^2 - \big(\frac{{F}^-_x}{{F}^-} + \frac{{G}^+_x}{{G}^+}\big)_{x} = \lambda_1(t),
	\\
-\ii\big(\frac{{F}^+_t}{{F}^+} - \frac{{G}^-_t}{{F}^-} \big) 
- \big(\bar{v}+\ii\frac{{G}^-_x}{{G}^-}- \ii\frac{{F}^+_x}{{F}^+}\big)^2 
+ \big(\frac{{F}^+_x}{{F}^+}+ \frac{{G}^-_x}{{G}^-}\big)_{x} =- \lambda_2(t).
\end{cases}
\end{align*}
Simplification shows that the first equation can be rewritten as
\begin{align*}
\ii\bigg(\frac{{F}^-_t}{{F}^-} - \frac{{G}^+_t}{{G}^+}\bigg) +2\ii \bar{u}\bigg(  \frac{F^-_x}{F^-}-\frac{G^+_x}{G^+}  \bigg)   -\bigg( \frac{{F}^-_{xx}}{{F}^-} + \frac{2{F}^-_x{G}^+_x}{{F}^-{G}^+} - \frac{{G}^+_{xx}}{{G}^+} \bigg)= \lambda_1-\bar{u}^2
\end{align*}
i.e.,
\begin{subequations}\label{2ILW3_elliptic}
\begin{align}\label{2ILW3a_elliptic}
\frac{(\ii D_t +2\ii \bar{u} D_x- D_x^2){F}^- \cdot {G}^+}{{F}^-{G}^+} = \lambda_1-\bar{u}^2.
\end{align}
In the same way, the second equation can be written as
\begin{align}\label{2ILW3b_elliptic}
\frac{(-\ii D_t + 2\ii \bar{v} D_x+D_x^2){F}^+\cdot{G}^-}{{F}^+{G}^-} = -\lambda_2+\bar{v}^2.
\end{align}
\end{subequations}
Multiplying (\ref{2ILW3a_elliptic}) and (\ref{2ILW3b_elliptic}) by ${F}^-{G}^+$ and ${F}^+{G}^-$, respectively, we conclude that (\ref{2ilw}) is equivalent to the bilinear system (\ref{hirota_form_elliptic}).
This completes the proof. 
\end{proof}

\subsection{Proof of Theorem \ref{bilinearth_elliptic}B}
\begin{proof}
We decompose the $L$-periodic solution $u,v$ of \eqref{2ilw} as $u=\bar{u}+u_++u_-$, $v=\bar{v}+v_++v_-$, with $u_{\pm}$, $v_{\pm}$ as in \eqref{uvprojections_elliptic}. We view \eqref{FG_to_uv_elliptic} as a pair of differential-difference equations for $F$, $G$ and seek solutions satisfying
\begin{equation}\label{implicit_ansatz_elliptic}
\begin{split}
&\ii \partial_z \log F(x-\ii\delta/2,t)=u_+(x,t),\qquad \ii\partial_z \log F(x+\ii\delta/2,t)=-v_+(x,t), \\
&\ii \partial_z \log G(x-\ii\delta/2,t)=v_-(x,t),\qquad \ii\partial_z \log G(x+\ii\delta/2,t)=-u_-(x,t),
\end{split}
\end{equation}
so that
\begin{equation}\label{RHsystem_elliptic}
\begin{split}
&\ii \partial_z \log F(x-\ii\delta/2,t)- \ii\partial_z\log F(x-\ii\delta/2,t)=u_{+}(x,t)+v_{+}(x,t) \\
&\ii \partial_z \log G(x-\ii\delta/2,t)- \ii\partial_z\log G(x+\ii\delta/2,t)=v_{-}(x,t)+u_{-}(x,t).
\end{split}
\end{equation}

Let $\tilde{\Pi}$ denote the torus $\C/\tilde{\Lambda}$, where $\tilde{\Lambda}\coloneqq L\Z+\ii\delta\Z$, $\tilde{\pi}$ the natural projection $\C\to \tilde{\Pi}$, and $\tilde{\Pi}_0$ the image of $\Im z=0$ under $\tilde{\pi}$. Then \eqref{RHsystem_elliptic} defines a pair of RH problems for the functions $\partial_z\log F$ and $\partial_z\log G$ on $\tilde{\Pi}$. The following lemma can be proved similarly to Lemma \ref{RHlemma_elliptic}.

\begin{lemma}[RH problem on $\tilde{\Pi}$ with a jump across $\tilde{\Pi}_0$]\label{RHlemma2_elliptic}
Let $J: \tilde{\Pi}_0 \to \C$ be a continuous function such that
$$\int_{-L/2}^{L/2} J(x)\, \mathrm{d}x=0.$$ 
Then the scalar RH problem: 
\begin{itemize}
\item $A:\tilde{\Pi} \setminus \tilde{\Pi}_0$ is analytic, 
\item across $\tilde{\Pi}_0$, $A$ satisfies the jump condition
$$A^+(x) - A^-(x) = J(x), \qquad x \in [-L/2,L/2),$$
\end{itemize}
has the general solution
\begin{align}\label{rh_soln2_elliptic}
A(z) =\frac{1}{2\pi\ii} \int_{-L/2}^{L/2} \zeta_1(x'-z|L/2,\ii\delta/2)J(x') \, \mathrm{d}x'+A_0, \qquad z \in \tilde{\Pi} \setminus \tilde{\Pi}_0,
\end{align}
where $A_0$ is an arbitrary complex constant. Moreover, this solution satisfies
$$A^\pm(x) = 
\frac{(T_{\frac{\delta}{2}}J)(x)}{2\ii} \pm \frac{1}{2} J(x), \qquad x \in  [-L/2,L/2),
$$
where
\begin{equation}\label{Thalfdelta_elliptic}
(T_{\frac{\delta}{2}}f)(x)\coloneqq\frac1{\pi}\pvint_\R \zeta_1(x'-x|L/2,\ii\delta/2)f(x')\,\mathrm{d}x'.
\end{equation}
\end{lemma}
\begin{remark}
Note that $T_{\frac{\delta}{2}}$ is the operator $T$ but with $\delta$ replaced by $\frac{\delta}{2}$: $T_{\frac{\delta}{2}}=\left.T\right|_{\delta\to \frac{\delta}{2}}.$
\end{remark}

It follows from \eqref{uvprojections_elliptic} and the anti-self-adjointness of $T$ and $\tilde{T}$ from Proposition \ref{TpropertiesS}, that the functions $u_{\pm}$, $v_{\pm}$ and hence $u_{\pm}+v_{\pm}$ have zero mean. Then, Lemma \ref{RHlemma2_elliptic} shows that the general solution to the scalar RH problem: 
\begin{itemize}
\item $A:\tilde{\Pi} \setminus \tilde{\Pi}_0$ is analytic, 
\item across $\tilde{\Pi}_0$, $A$ satisfies the jump condition
$$A^+(x) - A^-(x) = u_{\pm}+v_{\pm}, \qquad x \in [-L/2,L/2), $$
\end{itemize}
is given by
\begin{equation}
A(z)=\frac{1}{2\pi\ii}\int_{-L/2}^{L/2} \zeta_1(z-x|L/2,\ii\delta/2) \big(u_{\pm,x}(x')+v_{\pm,x}(x')\big)\,\mathrm{d}x'+A_0
\end{equation}
with boundary values
\begin{equation}
A^{\pm}(x)=\frac{T_{\frac{\delta}{2}}[u_{\pm}+v_{\pm}](x)}{2\ii}\pm\frac12\big(u_{\pm}(x)+v_{\pm}(x)\big)+A_0.
\end{equation}

Hence we find
\begin{align*}
\begin{dcases}
\ii\partial_z\log F(z,t)=\frac{1}{2\pi\ii}\int_{-L/2}^{L/2} \zeta_1(x'-z|L/2,\ii\delta/2)\big(u_{+}(x')+v_{+}(x')\big)\,\mathrm{d}x'+F_0(t), \\
\ii\partial_z\log G(z,t)=\frac{1}{2\pi\ii}\int_{-L/2}^{L/2}  \zeta_1(x'-z|L/2,\ii\delta/2)\big(u_{-}(x')+v_{-}(x')\big)\,\mathrm{d}x'+G_0(t),
\end{dcases}
\end{align*}
with corresponding boundary values
\begin{align*}
\begin{dcases}
\ii\partial_z\log F(x\pm\ii\delta/2,t)=\frac{T_{\frac{\delta}{2}}\big[u_{+}+v_{+}\big](x)}{2\ii}\pm\frac12\big(u_{+}(x)+v_{+}(x)\big)+F_0(t), \\
\ii\partial_z \log G(x\pm\ii\delta/2,t)=\frac{T_{\frac{\delta}{2}}\big[u_{-}+v_{-}\big](x)}{2\ii}\pm\frac12\big(u_{-}(x)+v_{-}(x)\big)+G_0(t), 
\end{dcases}
\end{align*}
where $F_0$ and $G_0$ are arbitrary complex functions of $t$.

\begin{lemma}\label{Thalfdetlalemma}
The functions $u_{\pm}$ and $v_{\pm}$ defined in \eqref{uvprojections_elliptic} satisfy
\begin{equation}\label{uvplusminushalfT_elliptic}
T_{\frac{\delta}2}(u_{\pm}+v_{\pm})=\pm\ii(u_{\pm}-v_{\pm}).
\end{equation}
\end{lemma}
\begin{proof}

We begin by considering the function
\begin{equation}\label{g_difference}
g(z)\coloneqq\wp(z|L/2,\ii\delta/2)-\wp(z|L/2,\ii\delta)-\wp(z+\ii\delta|L/2,\ii\delta).
\end{equation}
We note that $g(z)$ is doubly-periodic with periods $L$ and $\ii\delta$ and bounded. Hence, by Liouville's theorem, $g(z)$ is constant. Integrating \eqref{g_difference} and using the definition of the Weierstrass $\zeta_1$-function \eqref{zeta_j}, we see that
\begin{equation}
\zeta_1(z|L/2,\ii\delta/2)-\zeta_1(z|L/2,\ii\delta)-\zeta_1(z+\ii\delta|L/2,\ii\delta)=\alpha +\beta z,
\end{equation}
for some constants $\alpha,\beta\in\C$, but $L$-periodicity implies $\beta=0$. It follows from \eqref{Thalfdelta_elliptic} that 
$T_{\frac{\delta}{2}}$ can be written as $T_{\frac{\delta}{2}}=T+\tilde{T}$ on zero-mean functions. Using \eqref{uvprojections_elliptic}, we write
\begin{align*}
T_{\frac{\delta}2}(u_{\pm}+v_{\pm})=&\; \frac12(T+\tilde{T})\big(u+v-(\bar{u}+\bar{v})\mp \ii T(u-v-(\bar{u}-\bar{v}))\pm \ii\tilde{T}(u-v-(\bar{u}-\bar{v}))\big) \\
=&\; \frac12(T+\tilde{T})(u+v-(\bar{u}+\bar{v}))\mp\frac{\ii}{2}(T+\tilde{T})(T-\tilde{T})(u-v-(\bar{u}-\bar{v})).
\end{align*}
The identities $\tilde{T}Tf=T\tilde{T}f$ and $\tilde{T}\tilde{T}f=TTf+f-2\bar{f}$ from Proposition \ref{TpropertiesS} imply the identity
\begin{equation*}
(T+\tilde{T})(T-\tilde{T})f=-f+2\bar{f}.
\end{equation*}
Thus,
\begin{align*}
T_{\frac{\delta}2}(u_{\pm}+v_{\pm})=&\frac12(T+\tilde{T})(u+v-(\bar{u}+\bar{v}))\pm\frac{\ii}{2}(u-v-(\bar{u}-\bar{v})) \\
=&\pm\ii(u_{\pm}-v_{\pm}),
\end{align*}
which is \eqref{uvplusminushalfT_elliptic}. 
\end{proof}

Using Lemma \ref{Thalfdetlalemma}, we see that \eqref{implicit_ansatz_elliptic} is satisfied when $F_0(t)=0$ and $G_0(t)=0$, which gives \eqref{FG_from_uv_elliptic}.
Thus, \eqref{FG_to_uv_elliptic} holds and Theorem \ref{bilinearth_elliptic}A shows that \eqref{hirota_form_elliptic} with \eqref{lambdavalues} is satisfied. 
\end{proof}

\section{Periodic solitons}\label{solitonsec}

We construct the $N$-periodic soliton solutions of \eqref{2ilw} via an ansatz for the Hirota form \eqref{hirota_form_elliptic}. The ansatz
\begin{align}\label{sigma1ansatz}
F(x,t)=\prod_{j=1}^N e^{-\ii\pi x/L}\sigma_1(x-z_j(t)|L/2,\ii\delta)      ,\qquad G(x,t)=\prod_{j=1}^N e^{2\ii\pi x/L}\sigma_1(x-w_j(t)|L/2,\ii\delta),
\end{align}
for \eqref{hirota_form_elliptic}, together with Theorem \ref{bilinearth_elliptic}A, leads to an alternative proof of the following result in \cite{berntson2020a}. The na\"ive ansatz
\begin{align*}
F(x,t)=\prod_{j=1}^N \sigma_2(x-z_j(t)|L/2,\ii\delta) ,\qquad G(x,t)=\prod_{j=1}^N \sigma_2(x-w_j(t)|L/2,\ii\delta),
\end{align*}
fails to satisfy the conditions of Theorem~\ref{bilinearth_elliptic}A; the ansatz in \eqref{sigma1ansatz} is a minor modification of the latter one which satisfies those conditions. 

\begin{proposition}[Soliton solutions of the periodic non-chiral ILW equation]\label{solitoncorollary_elliptic}
For an arbitrary non-negative integer $N$ and complex parameters $a_j,b_j$ $(j=1,\ldots,N)$ satisfying
\begin{equation*}
\im(a_j\pm\ii\delta/2)\neq 2\delta n,\qquad \im(b_j\pm\ii\delta/2)\neq 2\delta n,
\end{equation*}
for all integers $n$, the functions
\begin{align}\label{uvpoleansatz}
\begin{dcases}
u(x,t)=\ii\sum_{j=1}^N\zeta_2 (x-z_j(t)-\ii\delta/2|L/2,\ii\delta)-\ii \sum_{j=1}^N \zeta_2(x-w_j(t)+\ii\delta/2|L/2,\ii\delta), \\
v(x,t)=-\ii \sum_{j=1}^N \zeta_2(x-z_j(t)+\ii\delta/2|L/2,\ii\delta) +\ii\sum_{j=1}^N \zeta_2(x-w_j(t)-\ii\delta/2|L/2,\ii\delta)  
\end{dcases}
\end{align}
provide a solution of the non-chiral ILW equation \eqref{2ilw} provided the poles $z_j(t)$ and $w_j(t)$ satisfy
\begin{align}\label{calogeroelliptic}
\begin{dcases}
\ddot{z}_j=-4\underset{k\neq j}{\sum_{k=1}^N}\wp'(z_j-z_k|L/2,\ii\delta),\qquad \im(z_j\pm\ii\delta/2)\neq 2\delta n,\\
\ddot{w}_j=-4\underset{k\neq j}{\sum_{k=1}^N}\wp'(w_j-w_k|L/2,\ii\delta),\qquad \im(w_j\pm\ii\delta/2)\neq 2\delta n,
\end{dcases}
\end{align}
with initial conditions
\begin{subequations}\label{calogeroinitialconditions}
\begin{align}
&z_j(0)=a_j,\quad w_j(0)=b_j, \label{calogeroinitialconditions1}\\
&\begin{dcases}
\dot{z}_j(0)=2\ii\underset{k\neq j}{\sum_{k=1}^N}\zeta_2(a_j-a_k|L/2,\ii\delta)-2\ii\sum_{k=1}^N \zeta_2(a_j-b_k+\ii\delta|L/2,\ii\delta), \\
\dot{w}_j(0)=-2\ii\underset{k\neq j}{\sum_{k=1}^N} \zeta_2(b_j-b_k|L/2,\ii\delta)+2\ii\sum_{k=1}^N \zeta_2(b_j-a_k+\ii\delta|L/2,\ii\delta).
\end{dcases}\label{calogeroinitialconditions2}
\end{align}
\end{subequations}
\end{proposition}

\begin{remark} 
As will become clear in the proof of Proposition~\ref{solitoncorollary_elliptic}, the pole ansatz in \eqref{uvpoleansatz} provides a solution of the periodic ncILW equation provided that
\begin{align}\label{calogerobacklund}
&\begin{dcases}
\dot{z}_j=2\ii\underset{k\neq j}{\sum_{k=1}^N}\zeta_2(z_j-z_k|L/2,\ii\delta)-2\ii\sum_{k=1}^N \zeta_2(z_j-w_k+\ii\delta|L/2,\ii\delta), \\
\dot{w}_j=-2\ii\underset{k\neq j}{\sum_{k=1}^N} \zeta_2(w_j-w_k|L/2,\ii\delta)+2\ii\sum_{k=1}^N \zeta_2(w_j-z_k+\ii\delta|L/2,\ii\delta).
\end{dcases}
\end{align}
Our result is obtained by the observation that the equations in \eqref{calogerobacklund} are a B\"acklund transformations for the elliptic CM system, i.e., if \eqref{calogerobacklund} is fulfilled, then \eqref{calogeroelliptic}
is implied; to keep this paper self-contained, we also give the proof of this known fact \cite{rauch1982}. 
\end{remark}

\subsection{Proof of Proposition \ref{solitoncorollary_elliptic}}
We prove that the ansatz \eqref{sigma1ansatz} inserted into \eqref{hirota_form_a_elliptic} implies the equations of motion \eqref{calogeroelliptic} with initial conditions \eqref{calogeroinitialconditions}. The analogous proof for \eqref{hirota_form_b_elliptic} is similar and hence omitted.
We divide \eqref{hirota_form_a_elliptic} by $F^-G^+$ to obtain 
\begin{align}\label{hirota_verification_1}
\lambda_1-\bar{u}^2=&\; \ii\partial_t(\log F^- - \log G^+)+2(\partial_x\log F^-)(\partial_x \log G^+)-\partial_x^2 (\log F^-+\log G^+) \nonumber\\
&\;-(\partial_x \log F^-)^2-(\partial_x\log G^+)^2+2\ii\bar{u}\partial_x(\log F^--\log G^+)
\end{align}
after using the identity
\begin{equation}
\frac{F^-_{xx}}{F^-}=\partial_x^2 \log F^-+(\partial_x\log F^-)^2
\end{equation}
and similarly for $G^+$. From \eqref{sigma1ansatz} we compute
\begin{align}\label{FminusGplusderivatives}
&\partial_t \log F^-=-\sum_{j=1}^N \zeta_1(x-z_j-\ii\delta/2) \dot{z}_j,\qquad \partial_t \log G^+=-\sum_{j=1}^N \zeta_1(x-w_j+\ii\delta/2)\dot{w}_j, \nonumber\\
&\partial_x \log F^-=-N\ii \frac{\pi}{L}+\sum_{j=1}^N \zeta_1 (x-z_j-\ii\delta/2),\qquad \partial_x \log G^+=-N\ii \frac{\pi}{L}+\sum_{j=1}^N \zeta_1 (x-w_j+\ii\delta/2), \nonumber\\
& \partial_x^2 \log F^-=-\sum_{j=1}^N \wp_1(x-z_j-\ii\delta/2)      ,\qquad \partial_x^2 \log G^+=-\sum_{j=1}^N \wp_1(x-w_j+\ii\delta/2). 
\end{align}
Substituting these into \eqref{hirota_verification_1} gives
\begin{align*}
\lambda_1-\bar{u}^2= &-\ii\sum_{j=1}^N \big(  \zeta_1(x-z_j-\ii\delta/2)\dot{z}_j-\zeta_1(x-w_j+\ii\delta/2)\dot{w}_j           \big)\\
&\;+ 2\Bigg(-N\ii \frac{\pi}{L}+\sum_{j=1}^N \zeta_1 (x-z_j)\Bigg)\Bigg( -N\ii \frac{\pi}{L}+\sum_{j=1}^N \zeta_1 (x-w_j)       \Bigg) \\
&\;+\sum_{j=1}^N \big( \wp_1(x-z_j-\ii\delta/2)+\wp_1(x-w_j+\ii\delta/2)\big)\\
&\; -\Bigg(-N\ii \frac{\pi}{L}+\sum_{j=1}^N \zeta_1 (x-z_j-\ii\delta/2)\Bigg)^2- \Bigg(-N\ii \frac{\pi}{L}+\sum_{j=1}^N \zeta_1 (x-z_j-\ii\delta/2)\Bigg)^2 \\
&\; +2\ii\bar{u}\sum_{j=1}^N \big(  \zeta_1(x-z_j-\ii\delta/2)-\zeta_1(x-w_j+\ii\delta/2)     \big),
\end{align*}
which becomes\footnote{Below we use shorthand notation for sums: $\sum_{k\neq j}^N=\sum_{k=1,k\neq j}^N$, etc.}
\begin{align}\label{hirotacalc1}
\lambda_1-\bar{u}^2=& -\ii \sum_{j=1}^N \big( \zeta_1(x-z_j-\ii\delta/2)\dot{z}_j-\zeta_1(x-w_j+\ii\delta/2)\dot{w}_j\big)  \nonumber\\
&\; +2\sum_{j=1}^N\sum_{k=1}^N \zeta_1(x-z_j-\ii\delta/2)\zeta_1(x-w_k+\ii\delta/2)  \nonumber\\
&\; +\sum_{j=1}^{N}\big( \wp_1(x-z_j-\ii \delta/2)+ \wp_1(x-w_j+\ii \delta/2)\big)  \nonumber\\
&\; -\sum_{j=1}^N \big( \zeta_1(x-z_j-\ii\delta/2)^2+\zeta_1(x-w_j+\ii\delta/2)^2\big)  \nonumber\\
&\;- \sum_{j=1}^N\sum_{k\neq j}^N  \zeta_1(x-z_j-\ii\delta/2)\zeta_1(x-z_k-\ii\delta/2) \nonumber\\
&\;- \sum_{j=1}^N\sum_{k\neq j}^N  \zeta_1(x-w_j+\ii\delta/2)\zeta_1(x-w_k+\ii\delta/2)\big) \nonumber\\
&\;+2\ii\bar{u}\sum_{j=1}^N \big(\zeta_1(x-z_j-\ii\delta/2)-\zeta_1(x-w_j+\ii\delta/2)\big)
\end{align}
after simplification. To proceed, it is useful to introduce the notation
\begin{equation}
\label{absr}
(Z_j,r_j)=\begin{cases}
(z_j+\ii\delta/2,+), & j=1,\ldots,N, \\
(w_{j-N}-\ii\delta/2,-), & j=N+1,\ldots 2N,
\end{cases}
\end{equation}
so that \eqref{hirotacalc1} can be written as
\begin{align*}
\lambda_1-\bar{u}^2=& -\ii \sum_{j=1}^{2N} r_j\zeta_1(x-Z_j)\dot{z}_j - \sum_{j=1}^{2N} \zeta_1(x-Z_j)^2 \\
&\; -   \sum_{j=1}^{2N}\sum_{k\neq j}^{2N}r_jr_k \zeta_1(x-Z_j)  \zeta_1(x-Z_k) \\
&\; +\sum_{j=1}^{2N} \wp_1(x-Z_j)+2\ii\bar{u} \sum_{j=1}^{2N} r_j \zeta_1(x-Z_j).
\end{align*}
Straightforward calculation using the identities (\ref{calogeroidentity1}-\ref{calogeroidentity2}) then establish that
\begin{align*}
\lambda_1-\bar{u}^2=&\; \sum_{j=1}^{2N} r_j\zeta_1(x-Z_j)\Bigg(-\ii \dot{Z}_j-2\sum_{k\neq j}^{2N} r_k\alpha_1(Z_j-Z_k)+2\ii\bar{u}\Bigg)\\
&\; -\frac12 \sum_{j=1}^{2N}\sum_{k\neq j}^{2N} r_jr_k f_1(Z_j-Z_k)-2N \frac{3\ii\eta_1}{2\delta},
\end{align*}
and setting 
\begin{equation}
\lambda_1(t)\equiv\bar{u}^2-\frac12 \sum_{j=1}^{2N}\sum_{k\neq j}^{2N} r_jr_k f_1(Z_j-Z_k)-2N \frac{3\ii\eta_1}{2\delta},
\end{equation}
we obtain the equations of motion 
\begin{equation}\label{Zjeom}
\dot{Z}_j=2\ii\sum_{k\neq j}^{2N} r_k\zeta_1(Z_j-Z_k)+2\bar{u},\qquad j=1,\ldots,2N,
\end{equation}
or, for $j=1,\ldots,N$,
\begin{equation}\label{zjwjeom}
\begin{dcases}\dot{z}_j=\ii \sum_{k\neq j}^N \zeta_1(z_j-z_k)-2\ii \sum_{k=1}^{N} \zeta_1(z_j-w_k-\ii\delta)+2\bar{u}, \\
\dot{w}_j= -2\ii \sum_{k\neq j}^N \zeta_1(w_j-w_k)+2\ii \sum_{k=1}^{N} \zeta_1(w_j-z_k+\ii\delta)+2\bar{u}.
\end{dcases}
\end{equation}

\begin{lemma}
The quantity
\begin{equation}\label{Xdefinition}
X\coloneqq \sum_{j=1}^N z_j-\sum_{j=1}^N w_j
\end{equation}
is conserved under the evolution of \eqref{zjwjeom}. 
\end{lemma} 

\begin{proof} 
We differentiate \eqref{Xdefinition} with respect to $t$ and insert \eqref{zjwjeom}:
\begin{align*}
\dot{X}=&\;\sum_{j=1}^N \dot{z}_j-\sum_{j=1}^N \dot{w}_j \\
=&\; \sum_{j=1}^N \Bigg(2\ii \sum_{k\neq j}^N \zeta_1(z_j-z_k)-2\ii \sum_{k=1}^{N} \zeta_1(z_j-w_k-\ii\delta)+2\bar{u}\Bigg)  \\ 
&\;-\sum_{j=1}^N \Bigg(-2\ii \sum_{k\neq j}^N \zeta_1(w_j-w_k)+2\ii \sum_{k=1}^{N} \zeta_1(w_j-z_k+\ii\delta)+2\bar{u}\Bigg) \\
=&\; 2\ii\sum_{j=1}^N \sum_{k\neq j}^N \big(\zeta_1(z_j-z_k)+\zeta_1(w_j-w_k)\big) \\
&\; - 2\ii\sum_{j=1}^N\sum_{k=1}^N \big(\zeta_1(z_j-w_k-\ii\delta)+\zeta_1(z_j-w_k+\ii\delta)\big)=0.
\end{align*}
\end{proof}

Using the identity
\begin{equation}
\zeta_1(z)-\zeta_2(z)=-\pi z/(L\delta),
\end{equation}
which follows from Definition~\ref{modifiedWeierstrassdefinition} and the identity $\eta_1\omega_2-\eta_2\omega_1=\ii\pi/2$ \cite[Eq. 23.2.14]{DLMF}, we can write 
\begin{align*}
\dot{z}_j=&\;2\ii \sum_{k\neq j}^N \zeta_2(z_j-z_k)-\frac{2\ii\pi}{L\delta}\sum_{k\neq j}^N \zeta_2(z_j-z_k) \\
&\; -2\ii \sum_{k=1}^{N} \zeta_2(z_j-w_k-\ii\delta)+\frac{2\ii\pi}{L\delta}\sum_{k=1}^N \zeta_2(w_j-z_k-\ii\delta)     +2\bar{u} \\
= &\; 2\ii\sum_{k\neq j}^{N} \zeta_2(z_j-z_k)-2\ii\sum_{k\neq j}^N \zeta_2(z_j-w_k-\ii\delta)+\frac{2\ii\pi}{L\delta}X+\frac{2\pi N}{L}+2\bar{u}
\end{align*}
and
\begin{align*}
\dot{w}_j=&\; -2\ii \sum_{k\neq j}^N \zeta_2(w_j-w_k)+ \frac{2\ii\pi}{L\delta}\sum_{k\neq j}^N \zeta_2(w_j-w_k) \\
&\;   + 2\ii \sum_{k=1}^{N} \zeta_2(w_j-z_k+\ii\delta)-\frac{2\ii\pi}{L\delta}\sum_{k=1}^N (w_j-z_k+\ii\delta) +2\bar{u}  \\
=&\; -2\ii\sum_{k\neq j}^N \zeta_2(w_j-w_k)+ 2\ii \sum_{k=1}^{N} \zeta_2(w_j-z_k+\ii\delta)+\frac{2\ii\pi}{L\delta}X+\frac{2\pi N}{L}+2\bar{u} .
\end{align*}
We set
\begin{equation}\label{ubardefinition}
\bar{u}\equiv -\frac{\ii\pi}{L\delta}X-\frac{N\pi}{L},
\end{equation}
so that the equations of motion \eqref{zjwjeom} become
\begin{equation}\label{zjwjeom2}
\begin{dcases}\dot{z}_j=2\ii \sum_{k\neq j}^N \zeta_2(z_j-z_k)-2\ii \sum_{k=1}^{N} \zeta_2(z_j-w_k-\ii\delta), \\
\dot{w}_j= -2\ii \sum_{k\neq j}^N \zeta_2(w_j-w_k)+2\ii \sum_{k=1}^{N} \zeta_2(w_j-z_k+\ii\delta).
\end{dcases}
\end{equation}
Hence, after inserting \eqref{ubardefinition} with \eqref{Xdefinition} and \eqref{FminusGplusderivatives} into \eqref{FG_to_uv_elliptic}, we have
\begin{align*}
u(x,t)=& -\frac{\ii\pi}{L\delta}\sum_{j=1}^N \zeta_1(z_j-w_j)-\frac{N\pi}{L}  + \ii\Bigg(-N\ii\frac{\pi}{L}+ \sum_{j=1}^N \zeta_1(x-z_j-\ii\delta/2)\Bigg) \\ 
& -\ii\Bigg(-N\ii\frac{\pi}{L}+ \sum_{j=1}^N \zeta_1(x-w_j+\ii\delta/2)\Bigg) \\
=&\; \ii\sum_{j=1}^N\bigg( \zeta_1(x-z_j-\ii\delta/2) +\frac{\pi}{L\delta}(x-z_j-\ii\delta/2)     \bigg) \\
& - \ii\sum_{j=1}^N\bigg( \zeta_1(x-w_j+\ii\delta/2) +\frac{\pi}{L\delta}(x-w_j+\ii\delta/2)     \bigg)  \\
=&\; \ii\sum_{j=1}^N \zeta_2(x-z_j-\ii\delta/2)-\ii\sum_{j=1}^N \zeta_2(x-w_j+\ii\delta/2),
\end{align*}
which is the first equation in \eqref{uvpoleansatz}. The corresponding result for $v(x,t)$ in \eqref{uvpoleansatz} is established similarly. By Theorem \ref{bilinearth_elliptic}A, \eqref{uvpoleansatz} provides a solution to \eqref{2ilw} when \eqref{zjwjeom2} is satisfied. 

It remains to show that \eqref{zjwjeom2} with initial conditions \eqref{calogeroinitialconditions1} is equivalent to \eqref{calogeroelliptic} with the initial conditions \eqref{calogeroinitialconditions}. We write \eqref{Zjeom} as 
\begin{equation}\label{Zjeom2}
r_j\dot{Z}_j=2\ii\sum_{k\neq j}^{2N} r_j r_k\zeta_2(Z_j-Z_k), \qquad j=1,\ldots,2N,
\end{equation}
and claim that
\begin{align}\label{Zjcalc1}
r_j\ddot{Z}_j=&\;-2\partial_{Z_j} \sum_{k=1}^{2N} r_k \Bigg(\sum_{l\neq k}^{2N}  r_l \zeta_2(Z_k-Z_l)\Bigg)^2 
\end{align}
Indeed, by direct computation,
\begin{align}\label{Zjcalc2}
r_j\ddot{Z}_j=&\; 4\sum_{k=1}^{2N} r_k \sum_{l\neq k}^{2N} r_l \zeta_2(a_k-a_l) \sum_{m\neq k}^{2N} r_m \wp_2(Z_k-Z_m)(\delta_{jk}-\delta_{jm}) \nonumber\\
=&\; 4\sum_{l\neq j}^{2N} \sum_{m\neq j}^{2N} r_j r_l r_m \zeta_2(Z_j-Z_l)\wp_2(Z_j-Z_m)\nonumber\\
&\;-4\sum_{k\neq j}^{2N}\sum_{l\neq k}^{2N} r_j r_k r_l \zeta_2(Z_k-Z_l)\wp_2(Z_k-Z_j) \nonumber\\
=&\;  4\sum_{k\neq j}^{2N} r_jr_k\wp_2(Z_k-Z_l) \Bigg( \sum_{l\neq j}^{2N} \zeta_2(Z_j-Z_l) -\sum_{l\neq k}^{2N} \zeta_2(Z_k-Z_l)    \Bigg);
\end{align}
alternatively, by differentiating \eqref{Zjeom2} with respect to $t$ and inserting \eqref{Zjeom}, we obtain
\begin{align*}
r_j\ddot{Z}_j=&\;-2\ii\sum_{k\neq j}^{2N} r_jr_k \wp_2(Z_j-Z_k)(\dot{Z}_j-\dot{Z}_k) \\
=&\; 4\sum_{k\neq j}^{2N} r_jr_k\wp_2(Z_j-Z_k)\Bigg(\sum_{l\neq j}^{2N} r_l\zeta_2(Z_j-Z_l)-\sum_{l\neq k}^{2N} r_l\zeta_2(Z_k-Z_l)  \Bigg),\end{align*}
which is the last line in \eqref{Zjcalc2}.

To show that \eqref{Zjcalc1} implies \eqref{calogeroelliptic}, we write
\begin{align*}
\sum_{k=1}^{2N} r_k \Bigg(\sum_{l\neq k}^{2N} r_l\zeta_2(Z_k-Z_l)\Bigg)^2      =&\;\sum_{k=1}^{2N}\sum_{l\neq k}^{2N}\sum_{m\neq k}^{2N} r_k r_l r_m \zeta_2(Z_k-Z_l)\zeta_2(Z_k-Z_m) \\
=&\; \sum_{k=1}^{2N} \sum_{l\neq k}^{2N} r_k \zeta_2(Z_k-Z_l)^2 \\
& +\sum_{k=1}^{2N}\sum_{l\neq k}^{2N}\sum_{m\neq k,l}^{2N} r_k r_l r_m \zeta_2(Z_k-Z_l)\zeta_2(Z_k-Z_m).
\end{align*}
Then, a lengthy but straightforward computation using the identities (\ref{calogeroidentity1}-\ref{calogeroidentity2}) shows that
\begin{align*}
\sum_{k=1}^{2N} r_k \Bigg(\sum_{l\neq k}^{2N} r_l\zeta_2(Z_k-Z_l)\Bigg)^2 = \sum_{k=1}^{2N} \sum_{l\neq k}^{2N} r_k  \wp_2(Z_k-Z_l) ,
\end{align*}
which, upon comparison with the first line of \eqref{Zjcalc1}, implies \eqref{calogeroelliptic} after recalling the notation \eqref{absr}.

\section{B\"{a}cklund transformation}\label{backlundsec}

Suppose $(u,v)$ and $(\tilde{u},\tilde{v})$ are two solutions of \eqref{2ilw} with associated Hirota bilinear forms \eqref{hirota_form_elliptic} and
\begin{subequations}\label{hirota_F2G2_elliptic}
\begin{align}
&(\ii D_t+2\ii \bar{u} D_x-D_x^2-\tilde{\lambda}_1(t)+\bar{u}^2)\tilde{F}^-\cdot \tilde{G}^+=0,\label{hirota_F2G2_a_elliptic} \\
&(\ii D_t-2\ii \bar{v} D_x-D_x^2-\tilde{\lambda}_2(t)+\bar{v}^2)\tilde{F}^+\cdot \tilde{G}^-=0,\label{hirota_F2G2_b_elliptic}
\end{align}
\end{subequations}
respectively, where
\begin{align*}
&u=\bar{u}+\ii\partial_x\log \frac{F^+}{G^-},\qquad v=\bar{v}+\ii\partial_x\log \frac{G^-}{F^+}, \\
&\tilde{u}=\bar{u}+\ii\partial_x\log \frac{\tilde{F}^+}{\tilde{G}^-},\qquad \tilde{v}=\bar{v}+\ii\partial_x\log \frac{\tilde{G}^-}{\tilde{F}^+}.
\end{align*}
Then, in terms of the variables $F$, $G$, $\tilde{F}$, $\tilde{G}$, the B\"{a}cklund transformation of \eqref{2ilw} is given by 
\begin{subequations}
\label{backlund_hirota_elliptic}
\begin{align}
&(\ii D_t-2\ii(\alpha_1-\bar{u}) D_x-D_x^2-\lambda_1+\bar{u}^2)F^-\cdot \tilde{F}^-=0, \label{backlund_hirota_a_elliptic}\\
&(\ii D_t-2\ii(\alpha_1-\bar{u}) D_x-D_x^2-\tilde{\lambda}_1+\bar{u}^2)G^+\cdot \tilde{G}^+=0, \label{backlund_hirota_b_elliptic}\\
&(D_x+\ii\alpha_1)G^+\cdot \tilde{F}^-=\ii\beta_1 F^-\cdot \tilde{G}^+,  \label{backlund_hirota_c_elliptic}\\
&(\ii D_t-2\ii(\alpha_2+\bar{v}) D_x-D_x^2-\lambda_2+\bar{v}^2)F^+\cdot\tilde{F}^+=0,      \label{backlund_hirota_d_elliptic} \\
&(\ii D_t-2\ii(\alpha_2+\bar{v}) D_x-D_x^2-\tilde{\lambda}_2+\bar{v}^2)G^-\cdot\tilde{G}^-=0, \label{backlund_hirota_e_elliptic}\\
&(D_x+\ii\alpha_2)G^-\cdot\tilde{F}^+=\ii\beta_2 F^+\cdot \tilde{G}^-,\label{backlund_hirota_f_elliptic}
\end{align}
\end{subequations}
where $\alpha_1$, $\alpha_2$, $\beta_1$, $\beta_2$ are arbitrary functions of time and $\lambda_1(t)$, $\lambda_2(t)$, $\tilde{\lambda}_1(t)$, $\tilde{\lambda}_2(t)$ are complex functions fixed by the Hirota forms (\ref{hirota_form_elliptic}, \ref{hirota_F2G2_elliptic}). The following Proposition can be established similarly to \cite[Proposition 4.1]{berntson2021}.

\begin{proposition}[B\"{a}cklund transformation in terms of bilinear variables]
Suppose $(F,G)$ and $(\tilde{F},\tilde{G})$ satisfy the relations in \eqref{backlund_hirota_elliptic}. Then $(F,G)$ is a solution of \eqref{hirota_form_elliptic} if and only $(\tilde{F},\tilde{G})$ is a solution of \eqref{hirota_F2G2_elliptic}. 
\end{proposition}

To transform \eqref{backlund_hirota_elliptic} into a form written in the original variables, we introduce potential functions $U$, $V$, $\tilde{U}$, $\tilde{V}$ by
\begin{equation}\label{UVdef_elliptic}
\begin{split}
U\coloneqq\ii\log \frac{F^-}{G^+},\qquad V\coloneqq\ii\log\frac{G^-}{F^+}, \\
\tilde{U}\coloneqq\ii\log \frac{\tilde{F}^-}{\tilde{G}^+},\qquad \tilde{V}\coloneqq\ii\log\frac{\tilde{G}^-}{\tilde{F}^+} ,
\end{split}
\end{equation}
so that
\begin{equation*}
U_x=u-\bar{u}, \qquad V_x=v-\bar{v},\qquad \tilde{U}_x=\tilde{u}-\bar{u},\qquad \tilde{V}_x=\tilde{v}-\bar{v}. 
\end{equation*}

\begin{lemma}\label{lambda12lemma}
The functions $\lambda_1$ and $\lambda_2$ in \eqref{hirota_form_elliptic} satisfy 
\begin{equation}\label{lambda1minuslambda2}
\lambda_1-\lambda_2=\frac{2}{L}I_2+\frac1L\frac{\mathrm{d}}{\mathrm{d}t}\int_{-L/2}^{L/2} (U+V)\,\mathrm{d}x,
\end{equation}
where 
\begin{equation}I_2\coloneqq\int_{-L/2}^{L/2} \frac12(u^2-v^2)\,\mathrm{d}x\end{equation}
is a constant. 
\end{lemma}
\begin{proof}
We add the two equations in \eqref{hirota_integrated} and integrate over $[-L/2,L/2]$ to obtain
\begin{align*}
L(\lambda_1-\lambda_2)=&\ii\frac{\mathrm{d}}{\mathrm{d}t}\int_{-L/2}^{L/2} \log\frac{F^-G^-}{F^+G^+}\,\mathrm{d}x+\ii(u_+-u_-+v_+-v_-)\big\rvert^{x=L/2}_{x=-L/2} \\
&+\int_{-L/2}^{L/2} \big(u^2-v^2\big)\,\mathrm{d}x.
\end{align*}
The second term vanishes by periodicity. Using \eqref{UVdef_elliptic} we obtain \eqref{lambda1minuslambda2}. It remains to show that $I_2$ is a conservation law. This is verified by a calculation analogous to the direct verification of $I_2$ in \cite[Section 5.3]{berntson2021}, using the anti-self-adjointness of $T$ and $\tilde{T}$ from Proposition \ref{TpropertiesS}.
\end{proof}

Lemma \ref{lambda12lemma} motivates the definition of the $t$-potential functions
\begin{subequations}
\begin{align}\label{Lambda1}
\Lambda_{1}\coloneqq\frac2L\int_{-L/2}^{L/2} u^2\,\mathrm{d}x+\frac{1}{L}\int_{L/2}^{L/2} U\,\mathrm{d}x, \qquad \Lambda_{2}\coloneqq\frac2{L}\int_{-L/2}^{L/2}v^2\,\mathrm{d}x-\frac{1}{L}\int_{L/2}^{L/2} V\,\mathrm{d}x, \\
\tilde{\Lambda}_{1}\coloneqq\frac2L\int_{-L/2}^{L/2} \tilde{u}^2\,\mathrm{d}x+\frac{2}{L}\int_{L/2}^{L/2} \tilde{U}\,\mathrm{d}x, \qquad \tilde{\Lambda}_{2}\coloneqq\frac1{L}\int_{-L/2}^{L/2}\tilde{v}^2\,\mathrm{d}x-\frac{1}{L}\int_{L/2}^{L/2} \tilde{V}\,\mathrm{d}x, \label{Lambda2}
\end{align}
\end{subequations}
so that $(\Lambda_1-\Lambda_2)_t=\lambda_1-\lambda_2$ and $(\tilde{\Lambda}_1-\tilde{\Lambda}_2)_t=\tilde{\lambda}_1-\tilde{\lambda}_2$.

\begin{theorem}[B\"{a}cklund transformation for the periodic non-chiral ILW equation]\label{backlundth_elliptic}
Suppose the following relations hold:
  \begin{subequations}\label{uvbacklund_elliptic}
\begin{align}\label{uvbacklunda_elliptic}
& u = \frac{1 - e^{-W}}{\epsilon} - \ii P_- W_x - \frac{1}{2}\tilde{T}Z_x,
	\\ \label{uvbacklundb_elliptic}
& W_t= -\frac{2}{\epsilon}(1 - e^{-W})W_x - TW_{xx} - \tilde{T}Z_{xx}
 + W_xTW_x + W_x \tilde{T}Z_x,
	\\ \label{uvbacklundc_elliptic}
& v = -\frac{1- e^Z}{\epsilon} + \ii P_+Z_x + \frac{1}{2}\tilde{T}W_x,
	\\ \label{uvbacklundd_elliptic}
& Z_t= -\frac{2}{\epsilon} (1- e^{Z}) Z_x + TZ_{xx} + \tilde{T}W_{xx}
 + Z_xTZ_x + Z_x\tilde{T}W_x,
\end{align}
\end{subequations}
where
\begin{equation}\label{WZ_elliptic}
W=\ii(U-\tilde{U}-(\Lambda_1-\tilde{\Lambda}_1)), \qquad Z=\ii(V-\tilde{V}+(\Lambda_2-\tilde{\Lambda}_2)),
\end{equation}
and
\begin{equation}\label{Ppm_elliptic}
P_{\pm}\coloneqq-\frac12(\ii T\pm 1).
\end{equation} Then $(u,v)$ satisfy the periodic non-chiral ILW equation (\ref{2ilw}) if and only if $(\tilde{u}, \tilde{v})$ do.
\end{theorem}

\begin{proof}
Let us first rewrite \eqref{backlund_hirota_c_elliptic}. Dividing \eqref{backlund_hirota_c_elliptic} by $G^+\tilde{F}^-$ yields
$$\frac{{G}^+_{x}}{G^+} - \frac{\tilde{F}^-_{x}}{\tilde{F}^-} + \ii\alpha_1 = \ii\beta_1 \frac{F^- \tilde{G}^+}{\tilde{F}^- G^+},$$
i.e.,
\begin{align}\label{fgbacklundcrewrite_elliptic}
u_- + \tilde{u}_+ = -\alpha_1 + \beta_1 e^{-\ii(U - \tilde{U})},
\end{align}
where $u_{\pm},$ $v_{\pm}$ are defined in \eqref{uvplusminusdef_elliptic} and $\tilde{u}_{\pm},$ $\tilde{v}_{\pm}$ are defined analogously. 

\begin{lemma}\label{uPlemma_elliptic}
The following identities hold:
  \begin{align}
\begin{cases}
    u_+ = P_- u - \frac{\ii}{2}\tilde{T}v-\frac12 (\bar{u}-\bar{v}) \\
    u_- = -P_- u + \frac{\ii}{2}\tilde{T}v + u-\frac12(\bar{u}+\bar{v})
    \end{cases} \qquad
  \begin{cases}
    \tilde{u}_+ = P_- \tilde{u} - \frac{\ii}{2}\tilde{T}\tilde{v}-\frac12(\bar{u}-\bar{v}), \\
    \tilde{u}_- = -P_- \tilde{u} + \frac{\ii}{2}\tilde{T}\tilde{v} +\tilde{u}-\frac12(\bar{u}+\bar{v}).
    \end{cases}   
  \end{align}
\end{lemma}
\begin{proof}
By Lemma \ref{lemma2_elliptic},
$$Tu + \tilde{T}v + \ii u -\ii(\bar{u}-\bar{v}) = 2\ii u_+$$
and the expression for $u_+$ follows after simplification. The expression for $u_-$ then follows because $u =\bar{u}+ u_+ + u_-$.
The expressions for $\tilde{u}_\pm$ follow in the same way. 
\end{proof}

Utilizing Lemma \ref{uPlemma_elliptic}, equation (\ref{fgbacklundcrewrite_elliptic}) can be rewritten as
\begin{align}\label{Pminusutildeu_elliptic}
-P_- (u-\tilde{u}) + \frac{\ii}{2}\tilde{T}(v - \tilde{v}) + u =\bar{u} -\alpha_1+ \beta_1 e^{-\ii(U - \tilde{U})}.
\end{align}
Recalling \eqref{WZ_elliptic} and setting 
$$\alpha_1=\bar{u}-\frac1\epsilon,\qquad \beta_1 =-\frac{1}{\epsilon}e^{\ii(\Lambda_1-\tilde{\Lambda}_1)},$$
this yields 
\begin{align*}
u = \frac{1 - e^{-W}}{\epsilon} - \ii P_- W_x - \frac{1}{2}\tilde{T}Z_x,
\end{align*}
which is (\ref{uvbacklunda_elliptic}).

We next rewrite the $t$-parts (\ref{backlund_hirota_a_elliptic})-(\ref{backlund_hirota_b_elliptic}) of the B\"acklund transformation as
\begin{align*}
\big(\ii \partial_t - 2\ii(\alpha_1-\bar{u})\partial_x)\log\frac{F^-}{\tilde{F}^-} - \partial_x^2 \log(F^-\tilde{F}^-) - \bigg(\partial_x\log\frac{F^-}{\tilde{F}^-}\bigg)^2 - \lambda_1+\bar{u}^2 = 0,\\
\big(\ii \partial_t - 2\ii(\alpha_1-\bar{u})\partial_x)\log\frac{G^+}{\tilde{G}^+} - \partial_x^2 \log(G^+\tilde{G}^+) - \bigg(\partial_x\log\frac{G^+}{\tilde{G}^+}\bigg)^2 - \tilde{\lambda}_1 +\bar{u}^2= 0.
\end{align*}
Subtracting the second of these equations from the first gives
\begin{align}
\begin{split}
& \big(\ii \partial_t - 2\ii(\alpha_1-\bar{u})\partial_x)\bigg(\log\frac{F^-}{G^+} - \log \frac{\tilde{F}^-}{\tilde{G}^+}\bigg)
- \bigg(\log\frac{F^-}{G^+} + \log\frac{\tilde{F}^-}{\tilde{G}^+}\bigg)_{xx}
	\\ \label{ipartalt2ilambda_elliptic}
& - \bigg(\log\frac{F^-}{G^+}- \log\frac{\tilde{F}^-}{\tilde{G}^+}\bigg)_x\big(\log({F^-}{G^+}) - \log({\tilde{F}^-}{\tilde{G}^+})\big)_x =\lambda_1-\tilde{\lambda}_1. 
\end{split}
\end{align}
Multiplying by $\ii$ and using the definitions (\ref{uvplusminusdef_elliptic}), (\ref{UVdef_elliptic}), and \eqref{Lambda1} of $u_\pm$, $\tilde{u}_\pm$, $U, \tilde{U}$, and $\Lambda_1,\tilde{\Lambda}_1$, this becomes
\begin{align*}
& \big(\ii \partial_t - 2\ii(\alpha_1-\bar{u})\partial_x\big)\big(U - \tilde{U}-(\Lambda_1-\tilde{\Lambda}_1)\big) \\
&- (U + \tilde{U})_{xx} +\ii (U - \tilde{U})_x\big(u_+ - u_- - (\tilde{u}_+ - \tilde{u}_-)\big) = 0.
\end{align*}
Recalling \eqref{WZ_elliptic} and using Lemma \ref{lemma2_elliptic}, we find
\begin{align*}
& W_t - 2(\alpha_1-\bar{u}) W_x - (U + \tilde{U})_{xx}
 - \ii W_x\big(TU + \tilde{T}V - T\tilde{U} - \tilde{T}\tilde{V}\big)_x = 0.
\end{align*}
Equation (\ref{Pminusutildeu_elliptic}) can be written as
$$\frac{1}{2}(U + \tilde{U})_x =\bar{u} -\alpha_1 + \beta_1 e^{-W} - \frac{1}{2} TW_x - \frac{1}{2} \tilde{T}Z_x.$$
Using this relation to eliminate $(U + \tilde{U})_{xx}$, we arrive at 
\begin{align*}
& W_t - 2(\alpha_1-\bar{u}) W_x + 2\beta_1 W_x e^{-W} + TW_{xx} + \tilde{T}Z_{xx}
 - W_x\big(TW_x + \tilde{T}Z_x\big) = 0
\end{align*}
That is,
\begin{align*}
& W_t = -\frac{2}{\epsilon}(1 - e^{-W})W_x - TW_{xx} - \tilde{T}Z_{xx}
 + W_xTW_x + W_x \tilde{T}Z_x,
\end{align*}
which is (\ref{uvbacklundb_elliptic}).

We next rewrite the $x$-part (\ref{backlund_hirota_f_elliptic}). Dividing (\ref{backlund_hirota_f_elliptic}) by $G^-\tilde{F}^+$ yields
$$\frac{{G}^-_{x}}{G^-} - \frac{{F}^+_{x}}{\tilde{F}^+} + \ii\alpha_2 = \ii\beta_2 \frac{F^+ \tilde{G}^-}{\tilde{F}^+ G^-},$$
i.e.,
\begin{align}\label{fgbacklundfrewrite_elliptic}
v_-  + \tilde{v}_+ = \alpha_1 - \beta_2 e^{\ii(V - \tilde{V})}.
\end{align}

\begin{lemma}\label{vPlemma_elliptic}
The following identities hold:
  \begin{align}
\begin{cases}
    v_+ = -P_+ v + \frac{\ii}{2}\tilde{T}u-\frac12(\bar{v}-\bar{u}), \\
    v_- = P_+ v - \frac{\ii}{2}\tilde{T}u + v-\frac12(\bar{v}+\bar{u}),
    \end{cases} \qquad
  \begin{cases}
      \tilde{v}_+ = -P_+ \tilde{v} + \frac{\ii}{2}\tilde{T}\tilde{u}-\frac12(\bar{v}-\bar{u}), \\
    \tilde{v}_- = P_+ \tilde{v} - \frac{\ii}{2}\tilde{T}\tilde{u} + \tilde{v}-\frac12(\bar{v}+\bar{u}).
  \end{cases}   
  \end{align}
\end{lemma}
\begin{proof}
By Lemma \ref{lemma2_elliptic},
$$Tv + \tilde{T}u - \ii v+\ii(\bar{v}-\bar{u})= -2\ii v_+$$
and the expression for $v_+$ follows after simplification. The expression for $v_-$ then follows because $v =\bar{v}+ v_+ + v_-$.
The expressions for $\tilde{v}_\pm$ follow in the same way. 
\end{proof}

Utilizing Lemma \ref{vPlemma_elliptic}, equation (\ref{fgbacklundfrewrite_elliptic}) can be rewritten as
\begin{align}\label{Pminusvtildev_elliptic}
P_+ (v- \tilde{v}) - \frac{\ii}{2}\tilde{T}(u - \tilde{u}) + v
= \bar{v}+\alpha_2 - \beta_2 e^{\ii(V - \tilde{V})}.
\end{align}
With \eqref{WZ_elliptic} and setting
$$\alpha_2 =-\bar{v}-\frac1\epsilon     ,\qquad \beta_2 =-\frac{1}{\epsilon}e^{\ii(\Lambda_2-\tilde{\Lambda}_2)},$$
this becomes
\begin{align*}
v = -\frac{1- e^Z}{\epsilon} + \ii P_+Z_x + \frac{1}{2}\tilde{T}W_x,
\end{align*}
which is (\ref{uvbacklundc_elliptic}).

We next rewrite the $t$-parts (\ref{backlund_hirota_d_elliptic})-(\ref{backlund_hirota_e_elliptic}) of the B\"acklund transformation.
As before, we find that (\ref{ipartalt2ilambda_elliptic}) holds except that $F,\tilde{F}$ and $G,\tilde{G}$ are now evaluated at $x + \ii\delta/2$ and $x - \ii\delta/2$, respectively, i.e.,
\begin{align}\begin{split}
& \big(\ii \partial_t - 2\ii(\alpha_2+\bar{v})\partial_x)\bigg(\log\frac{{F^+}}{G^-} - \log \frac{\tilde{F}^+}{\tilde{G}^-}\bigg)
- \bigg(\log\frac{F^+}{G^-} + \log\frac{\tilde{F}^+}{\tilde{G}^-}\bigg)_{xx}
	\\\label{ipartalt2ilambda_2_elliptic}
& - \bigg(\log\frac{F^+}{G^-}- \log\frac{\tilde{F}^+}{\tilde{G}^-}\bigg)_x\big(\log({F^+}{G^-}) - \log({\tilde{F}^+}{\tilde{G}^-})\big)_x = \lambda_2-\tilde{\lambda}_2. 
\end{split}
\end{align}
Multiplying by $\ii$ and using the definitions (\ref{uvplusminusdef_elliptic}), (\ref{UVdef_elliptic}), and \eqref{Lambda2} of $v_\pm$, $\tilde{v}_\pm$,$V, \tilde{V}$, and $\Lambda_2,\tilde{\Lambda}_2$ this becomes
\begin{align*}
& \big(\ii \partial_t - 2\ii(\alpha_2+\bar{v})\partial_x\big)\big(-V + \tilde{V}-\Lambda_2-\tilde{\Lambda}_2\big)  \\
&+ \big(V + \tilde{V})_{xx}+ \ii (-V + \tilde{V}\big)_x\big(-v_+ + v_- + \tilde{v}_+ - \tilde{v}_-\big) =0.
\end{align*}
Recalling \eqref{WZ_elliptic} and using Lemma \ref{lemma2_elliptic}, we find
\begin{align*}
& -Z_t + 2(\alpha_2+\bar{v}) Z_x + (V + \tilde{V})_{xx}
 + \ii Z_x\big(TV + \tilde{T}U - T\tilde{V} - \tilde{T}\tilde{U}\big)_x = 0.
\end{align*}
Equation (\ref{Pminusvtildev_elliptic}) can be written as
$$\frac{1}{2} (V + \tilde{V})_x  
= \bar{v}+\alpha_2 - \beta_2 e^{Z} + \frac{1}{2} TZ_x + \frac{1}{2}\tilde{T}W_x.$$
Using this relation to eliminate $(V + \tilde{V})_{xx}$, we arrive at 
\begin{align*}
& -Z_t + 2(\alpha_2+\bar{v}) Z_x - 2\beta_2 Z_x e^{Z} + TZ_{xx} + \tilde{T}W_{xx}
 + Z_x\big(TZ_x + \tilde{T}W_x\big) = 0.
\end{align*}
That is,
\begin{align*}
Z_t= -\frac{2}{\epsilon} (1- e^{Z}) Z_x + TZ_{xx} + \tilde{T}W_{xx}
 + Z_xTZ_x + Z_x\tilde{T}W_x,
\end{align*}
which is (\ref{uvbacklundd_elliptic}).
This completes the proof of Theorem \ref{backlundth_elliptic}.
\end{proof}

\section{Conservation laws}

\begin{theorem}[Conservation laws of the periodic non-chiral ILW equation]\label{conservation_laws_elliptic}
The periodic non-chiral ILW equation \eqref{2ilw} with \eqref{TT_elliptic} has an infinite number of conservation laws
\begin{equation}
I_n=\int_{-L/2}^{L/2} (W_n+Z_n)\,\mathrm{d}x,
\end{equation}
where $W_n$ and $Z_n$ can be computed recursively from the formal power series in $\epsilon$
\begin{subequations}\label{WZimplicit_elliptic}
\begin{align}
& u = \frac{1 - \exp\bigg(-\sum\limits_{n=1}^\infty W_n\epsilon^n\bigg) }{\epsilon} -\ii P_- \sum_{n=1}^\infty W_{n,x} \epsilon^n - \frac{1}{2}\tilde{T}\sum_{n=1}^\infty Z_{n,x} \epsilon^n,
	\\
& v = -\frac{1- \exp\bigg(\sum\limits_{n=1}^\infty Z_{n} \epsilon^n\bigg)}{\epsilon} +\ii P_+\sum_{n=1}^\infty Z_{n,x} \epsilon^n  + \frac{1}{2}\tilde{T}\sum_{n=1}^\infty W_{n,x} \epsilon^n,
\end{align}
with $P_{\pm}$ as in \eqref{Ppm_elliptic}.
\end{subequations}
The first four conservation laws are 
\begin{subequations}\label{3conservation_laws_elliptic}
\begin{align}
I_1=\int_{-L/2}^{L/2}& (u+v)\,\mathrm{d}x, \label{conservation_law1_elliptic}\\
I_2=\int_{-L/2}^{L/2}&\frac12(u^2-v^2)\,\mathrm{d}x, \label{conservation_law2_elliptic}\\
I_3=\int_{-L/2}^{L/2}& \bigg( \frac13(u^3+v^3)+\frac12( uTu_x+vTv_x+u\tilde{T}v_x+v\tilde{T}u_x    \bigg)\mathrm{d}x. \label{conservation_law3_elliptic} \\
I_4=\int_{-L/2}^{L/2}& \bigg(\frac{u^4-v^4}{4}+\frac{u_x^2-v_x^2}{8}+\frac38\big( (Tu_x)^2-(Tv_x)^2-(\tilde{T}u_x)^2+(\tilde{T}v_x)^2\big) \label{conservation_law4_elliptic}\\
&+\frac34\big(u^2 Tu_x-v^2 Tv_x\big)+\frac34\big(u^2\tilde{T}v_x-v^2\tilde{T}u_x\big)\bigg)\mathrm{d}x.\nonumber
\end{align}
\end{subequations}
\end{theorem}
\begin{proof}
Adding equations (\ref{uvbacklundb_elliptic}) and (\ref{uvbacklundd_elliptic}), we find
\begin{align}\label{WtplusZt}
W_t + Z_t = & -\frac{2}{\epsilon}(1 - e^{-W})W_x - TW_{xx} - \tilde{T}Z_{xx}
 + W_xTW_x + W_x \tilde{T}Z_x
 	\\
& -\frac{2}{\epsilon} (1- e^{Z}) Z_x + TZ_{xx} + \tilde{T}W_{xx}
 + Z_xTZ_x + Z_x\tilde{T}W_x. \nonumber
\end{align}
Thus,
\begin{align*}
\frac{\mathrm{d}}{\mathrm{d}t} \int_{-L/2}^{L/2} (W + Z)\,\mathrm{d}x
=& \int_\R \big(W_xTW_x + W_x \tilde{T}Z_x + Z_xTZ_x + Z_x\tilde{T}W_x\big)\,\mathrm{d}x. 
\end{align*}
Using the anti-self-adjointness (\ref{anti_self_adjointTTtilde_elliptic}) of the operators $T$ and $\tilde{T}$, the integral on the right-hand side vanishes. The remainder of the proof is identical to that of \cite[Theorem 4]{berntson2021} and hence omitted.
\end{proof}

\appendix

\section{Elliptic functions}\label{app:elliptic}

\begin{definition}[Weierstrass functions]
Consider a pair of complex numbers $\omega_1,\omega_2$ satisfying $\im\,(\omega_2/\omega_1)> 0$. Let $\Lambda\coloneqq 2\omega_1\Z+2\omega_2\Z$. Then, the Weierstrass $\sigma$-function with half-periods $\omega_1,\omega_2$ is defined as 
\begin{equation}
\sigma(z|\omega_1,\omega_2)\coloneqq z \prod_{\lambda \in \Lambda\setminus\{0\}} \bigg(\bigg(1-\frac{z}{\lambda}\bigg)\exp\bigg(\frac{z}{\lambda}+\frac{z^2}{2\lambda^2} \bigg)\bigg),
\end{equation}
the Weierstrass $\zeta$-function is defined as
\begin{equation}
\zeta(z|\omega_1,\omega_2)\coloneqq \partial_z \log \sigma(z|\omega_1,\omega_2),
\end{equation}
and the Weierstrass $\wp$-function is defined as
\begin{equation}
\wp(z|\omega_1,\omega_2)\coloneqq -\partial_z \zeta(z|\omega_1,\omega_2).
\end{equation}

\end{definition}

It is convenient to define minor modifications of the Weierstrass functions with enhanced periodicity properties.

\begin{definition}[Modified Weierstrass functions]\label{modifiedWeierstrassdefinition}
The modified Weierstrass $\sigma$-functions are
\begin{equation}\label{sigma_j}
\sigma_j(z|\omega_1,\omega_2)=e^{-\eta_j z^2/2\omega_j}\sigma(z|\omega_1,\omega_2), \qquad j=1,2,
\end{equation}
the modified Weierstrass $\zeta$-functions are
\begin{equation}\label{zeta_j}
\zeta_j(z|\omega_1,\omega_2)=\partial_z\log \sigma_j(z|\omega_1,\omega_2)=\zeta(z|\omega_1,\omega_2)-\frac{\eta_j}{\omega_j}z, \qquad j=1,2,
\end{equation}
and the modified Weierstrass $\wp$-functions are
\begin{equation}\label{wp_j}
\wp_j(z|\omega_1,\omega_2)=-\partial_z\zeta_j(z|\omega_1,\omega_2)=\wp(z|\omega_1,\omega_1)+\frac{\eta_j}{\omega_j}, \qquad j=1,2,
\end{equation}
where $\eta_j\coloneqq \zeta(\omega_j|\omega_1,\omega_2)$ for $j=1,2$.
\end{definition}

\begin{proposition}\label{weierstrassprop}
The modified Weierstrass functions satisfy the following identities:
\begin{align}
&\sigma_j(z+2\omega_j)=-\sigma_j(z),\qquad j=1,2, \label{sigmashift1}\\
&\sigma_1(z+2\omega_2)=-e^{-\ii\pi(z+\omega_2)/\omega_1}\sigma_1(z), \quad \sigma_2(z+2\omega_1)=-e^{\ii\pi (z+\omega_1)/\omega_2}\sigma_2(z), \label{sigmashift2}\\
&\zeta_j(z+2\omega_j)=\zeta_j(z), \qquad j=1,2, \label{zetashift1}\\
&\zeta_1(z+2\omega_2)=\zeta_1(z)-\frac{\ii\pi}{\omega_1},\qquad \zeta_2(z+2\omega_1)=\zeta_2(z)+\frac{\ii\pi}{\omega_2}, \label{zetashift2}\\
&\wp_j(z+2\omega_j)=\wp_j(z),\qquad j,k=1,2. \label{wpshift}
\end{align}
\end{proposition}

\begin{proof}
(\ref{sigmashift1}) and (\ref{sigmashift2}). We write
\begin{align*}
\sigma_j(z+2\omega_k)=&e^{-\eta_j(z+2\omega_k)^2/2\omega_j}\sigma(z+2\omega_k).
\end{align*}
Using the identity \cite[Eq. 23.2.15]{DLMF} $\sigma(z+2\omega_j)=-e^{-2\eta_j(z+\omega_j)}\sigma(z)$, we have
\begin{align*}
\sigma_j(z+2\omega_k)=&-e^{-\eta_j(z+2\omega_k)^2/2\omega_j}e^{2\eta_k(z+\omega_k)}\sigma(z) \\
=&-e^{-\eta_j z^2/2\omega_j}e^{-2\eta_j(\omega_k z+\omega_k^2)/\omega_j}e^{2\eta_k(z+\omega_k)}\sigma(z) \\
=& -e^{2(z+\omega_k)(\eta_k-\omega_k\eta_j/\omega_j)}    \sigma_j(z) .
\end{align*}
When $k=j$, we immediately obtain \eqref{sigmashift1}. Otherwise, we use the identity \cite[Eq. 23.2.14]{DLMF} $\eta_1\omega_2-\eta_2\omega_1=\ii\pi/2$, so that
$$
\sigma_1(z+2\omega_2)=-e^{2(z+\omega_2)(-\ii\pi/2\omega_1)}\sigma_1(z), \qquad \sigma_2(z+2\omega_1)=-e^{2(z+\omega_1)(\ii\pi/2\omega_2)}\sigma_2(z),
$$
which is \eqref{sigmashift2}.
\hfill\break

(\ref{zetashift1}) and (\ref{zetashift2}). These follow from logarithmic differentiation of \eqref{sigmashift1} and \eqref{sigmashift2}, respectively. 
\hfill\break

\eqref{wpshift}. The functions $\wp_j(z)$ differ from $\wp(z)$ by constants \eqref{wp_j} and so retain double-periodicity. 
\end{proof}

\begin{proposition}
The modified Weierstrass functions satisfy the following identities:
\begin{align}
\zeta_j(z)^2=&\;\wp_j(z)+f_j(z),\qquad j=1,2, \label{calogeroidentity1}\\
\zeta_j(z-a)\zeta_j(z-b)=&\;\big(\zeta_j(z-a)-\zeta_j(z-b)\big)\zeta_j(a-b) \nonumber\\
&\; +\frac12(f_j(z-a)+f_j(z-b)+f_j(a-b))-3\eta_j/2\omega_j, \qquad j=1,2, \label{calogeroidentity2}
\end{align}
where 
\begin{equation}f_j(z)\coloneqq \frac{\sigma_j''(z)}{\sigma_j(z)}= \frac{\sigma''(z)}{\sigma(z)}-\frac{\eta_j}{\omega_j}\bigg(2z\zeta_j(z)+\frac{\eta_j}{\omega_j}z^2+1\bigg),\qquad j=1,2.
\end{equation}
\end{proposition}

\begin{proof}
\eqref{calogeroidentity1}. We recall the identity $\zeta(z)^2=\wp(z)+(\sigma''/\sigma)(z)$. Using (\ref{zeta_j}-\ref{wp_j}) to write $\zeta(z)=\zeta_j(z)+(\eta_j/\omega_j)z$ and $\wp(z)=\wp_j(z)-\eta_j/\omega_j$, we obtain the result after some algebra.

\hfill \break

\eqref{calogeroidentity2}. We start from the identity $(\zeta(x)+\zeta(y)+\zeta(z))^2=\wp(x)+\wp(y)+\wp(z)$, which is valid when $x+y+z=0$. We consider the particular case $(\zeta(z-a)-\zeta(z-b)+\zeta(a-b))^2=\wp(z-a)+\wp(z-b)+\wp(a-b)$ (where we have used the fact that $\zeta$ is an odd function). Again using $\zeta(z)=\zeta_j(z)+(\eta_j/\omega_j)z$ and $\wp(z)=\wp_j(z)-\eta_j/\omega_j$, we have
\begin{align*}
\big(\zeta_j(z-a)-\zeta_j(z-b)+\zeta_j(a-b)\big)^2=\wp_j(z-a)+\wp_j(z-b)+\wp_j(a-b)-\frac{3\eta_j}{\omega_j}.
\end{align*}
Rearranging, we have
\begin{align*}
\zeta_j(x-a)\zeta_j(x-b)=&\big(\zeta_j(z-a)-\zeta_j(z-b)\big)\zeta_j(a-b)+\frac12\big(\zeta_j(z-a)^2-\wp_j(z-a)\\
&+\zeta_j(z-b)^2-\wp_j(z-b)+\zeta_j(a-b)^2-\wp_j(a-b)\big)-\frac{3\eta_j}{2\omega_j} \\
=&\big(\zeta_j(z-a)-\zeta_j(z-b)\big)\zeta_j(a-b) \\
&+\frac12\big(f_j(z-a)+f_j(z-b)+f_j(a-b)\big)-\frac{3\eta_j}{2\omega_j},
\end{align*}
where we have used the previous result.
\end{proof}

\section{Properties of the $T$ and $\tilde{T}$ operators}\label{app:TT}

In this section we collect and prove several identities for the $T$ and $\tilde{T}$ operators \eqref{TT_elliptic}.

\begin{proposition}[Properties of $T$ and $\tilde{T}$ on the circle]\label{TpropertiesS}
The operators $T$ and $\tilde{T}$ defined in \eqref{TT_elliptic} have the following properties
\begin{align}
\label{Tderivative_elliptic}
& \partial_x (Tf)(x)=(T\partial_x f)(x),\quad \partial_x (\tilde{T}f)(x)=(\tilde{T}\partial_x f)(x),  \qquad  x \in [-L/2,L/2),  \\
\label{anti_self_adjointTTtilde_elliptic}
&\int_{-L/2}^{L/2} f (Tg)\,\mathrm{d}x=-\int_{-L/2}^{L/2} (Tf)g\,\mathrm{d}x,\quad \int_{-L/2}^{L/2} f (\tilde{T}g)\,\mathrm{d}x=-\int_{-L/2}^{L/2} (\tilde{T}f)g\,\mathrm{d}x, \\
\label{Tcommutator_elliptic}
&(\tilde{T} T f)(x) = (T\tilde{T}f)(x), \qquad  x \in [-L/2,L/2), \\
\label{TTcommutator_elliptic}
&(\tilde{T} \tilde{T} f)(x) =(TTf)(x)+f(x)-\frac{2}{L}\int_{-L/2}^{L/2}f(x)\,\mathrm{d}x,  \qquad x \in [-L/2,L/2),  \\
&T[1]=\frac{1}{\pi}\int_{-L/2}^{L/2} \zeta_1(x'-x)\,\mathrm{d}x'=0, \label{Tconstantint1} \\
&\tilde{T}[1]=\frac{1}{\pi}\int_{-L/2}^{L/2} \zeta_1(x'-x+\ii\delta)\,\mathrm{d}x'=-\ii, \label{Ttconstantint1} \\
&\int_{-L/2}^{L/2} \zeta_1(x+a)\,\mathrm{d}x =\begin{cases}
-\ii \pi, &   0< \im a<2\delta, \\
+\ii \pi, & -2\delta<\im a< 0 .
\end{cases}
\label{shiftedzetaint}
\end{align}
\end{proposition}
The proofs of (\ref{Tderivative_elliptic}-\ref{TTcommutator_elliptic}) are similar to those for the analogous properties of $T$ and $\tilde{T}$ on $\R$ \cite[Proposition~A.1]{berntson2021} and hence omitted. We prove (\ref{Tconstantint1}-\ref{shiftedzetaint}).

\begin{proof}
(\ref{Tconstantint1}) and (\ref{Ttconstantint1}). By the definition of $\zeta_1$,
\begin{align}\label{Tconstantint2}
T[1]=&\frac1\pi\pvint_{-L/2}^{L/2} \zeta_1(x'-x)\,\mathrm{d}x' \\
 =&\frac1\pi\pvint_{-L/2}^{L/2} \zeta(x'-x)\,\mathrm{d}x'-\frac{2\eta_1}{L}\frac{1}{\pi}\int_{-L/2}^{L/2} (x'-x)\,\mathrm{d}x'. \nonumber 
\end{align}
The first integral in \eqref{Tconstantint1} can be computed using the definition of the principal value integral and the standard elliptic identities \cite{DLMF} $\sigma(-z)=-\sigma(z)$, $\sigma(z+2\omega_1)=-e^{-2\eta_1(z+\omega_1)}\sigma(z)$:
\begin{align*}
\frac1\pi\pvint_{-L/2}^{L/2} \zeta(x'-x)\,\mathrm{d}x'=&\; \frac{1}{\pi}\lim\limits_{\epsilon\to 0^+}\Bigg( \int_{-L/2}^{x-\epsilon}+\int_{x+\epsilon}^{L/2}\Bigg) \zeta(x'-x)\,\mathrm{d}x' \\
=&\; \frac{1}{\pi}\lim\limits_{\epsilon\to 0^+} \Bigg(  \log |\sigma(x'-x)| \Big|^{x-\epsilon}_{-L/2}  + \log |\sigma(x'-x)|\Big|^{L/2}_{x+\epsilon}       \Bigg) \\
=&\; \frac{1}{\pi}  \log\Bigg|\frac{\sigma(x-L/2)}{\sigma(x+L/2)} \Bigg| \\
=&-\frac{2\eta_1 x}{\pi}.
\end{align*}
The second integral in \eqref{Tconstantint2} is found to be
\begin{align*}
\frac{2\eta_1}{L}\frac{1}{\pi}\int_{-L/2}^{L/2} (x'-x)\,\mathrm{d}x'=-\frac{2\eta_1 x}{\pi}.
\end{align*}
Hence, the right-hand side of \eqref{Tconstantint1} vanishes. 

The function $f(z)=1$ satisfies the conditions of Lemma \ref{lemma1_elliptic}, hence \eqref{Ttconstantint1} follows from \eqref{Tconstantint1} and \eqref{TTtildegplus_elliptic}. 
\hfill \break

\eqref{shiftedzetaint}. We consider the integral
\begin{equation*}
\oint_{\Gamma} \zeta_1(z+\ii\delta)\,\mathrm{d}z,
\end{equation*}
where $\Gamma$ is a rectangular contour with vertices at $\pm L/2$ and $\pm L/2+\ii (\im a-\delta)$, oriented so the integral along the real axis is positively-oriented. When $0<\im a<2\delta$, the contour encloses no poles, so we have, after cancelling vertical contributions by periodicity, 
\begin{align*}
0=\oint_{\Gamma} \zeta_1(z+\ii\delta)\,\mathrm{d}z=\int_{-L/2}^{L/2} \zeta_1(x+\ii\delta)\,\mathrm{d}x-\int_{-L/2+\ii (\im a-\delta)}^{L/2+\ii(\im a-\delta)}\zeta_1(z+\ii\delta)\,\mathrm{d}z.
\end{align*} 
Changing variables in the second integral, we have
\begin{align*}
0=\int_{-L/2}^{L/2} \zeta_1(x+\ii\delta)\,\mathrm{d}x-\int_{-L/2}^{L/2}\zeta_1(x+\ii\,\im a)\,\mathrm{d}x.
\end{align*} 
Now using \eqref{Ttconstantint1}, we find
\begin{equation}\label{zeta1plusimaint}
\int_{-L/2}^{L/2}\zeta_1(x+\ii\,\im a)\,\mathrm{d}x=-\ii\pi.
\end{equation}
The first case in \eqref{shiftedzetaint} follows from the real translation invariance of \eqref{zeta1plusimaint}. The proof of the second case in \eqref{shiftedzetaint} is similar after accounting for the pole enclosed by $\Gamma$. 
\end{proof}

\bigskip
\noindent
{\bf Acknowledgement} {\it We thank Junichi Shiraishi for inspiring discussions, and Rob Klabbers for valuable comments and discussions that helped us to improve this paper. BKB acknowledges support from the G\"oran Gustafsson Foundation. EL acknowledges support from the Swedish Research Council, Grant No. 2016-05167, and by the Stiftelse Olle Engkvist Byggm\"astare, Contract 184-0573.
JL is grateful for support from the G\"oran Gustafsson Foundation, the Ruth and Nils-Erik Stenb\"ack Foundation, the Swedish Research Council, Grant No. 2015-05430, and the European Research Council, Grant Agreement No. 682537.}
\bigskip

\nocite{kodama1982}
\nocite{lebedev1987}

\bibliographystyle{unsrt}

\bibliography{BLL4}

\end{document}